\documentclass{article} 
\usepackage{graphicx}
\usepackage{krish} 
\usepackage{tikz}
\usepackage{stmaryrd}
\usepackage[most]{tcolorbox}
\usepackage{marginnote}
\usepackage{libertine}
\usepackage[libertine]{newtxmath}

\usetikzlibrary{arrows.meta,bending}

\newcommand{\bnntime}{\tau_{\calM}}

  {\endMakeFramed}

\algdef{SE}[SUBALG]{Indent}{EndIndent}{}{\algorithmicend\ }%
\algtext*{Indent}
\algtext*{EndIndent}

\title{Fast Geometric Spanners via Approximate Nearest Neighbor Search} 
\author{ Alexandr Andoni \\[1mm] Columbia University \\ \href{andoni@cs.columbia.edu}{andoni@cs.columbia.edu} \and 
Manuel Paez \\[1mm] Columbia University \\\href{map2332@seas.upenn.edu}{map2332@seas.upenn.edu} \and 
Krish Singal \\[1mm]University of Pennsylvania \\ \href{ksingal@seas.upenn.edu}{ksingal@seas.upenn.edu}
  \and 
Tian Zhang \\[1mm]University of Pennsylvania \\ \href{tianzh@seas.upenn.edu}{tianzh@seas.upenn.edu}
  \and 
}

\begin{document} 

\maketitle 


\begin{abstract}
    We study the problem of constructing metric spanners in \emph{general} metric spaces in subquadratic time when given blackbox access to a fast algorithm for batch approximate nearest neighbor search. In particular, we show the following results for \emph{any} metric space $\calM = ([n], \sfd)$ with aspect ratio $\Delta$ admitting a $c$-approximate batch nearest neighbor search algorithm with runtime $\bnntime(n)$, 
    \begin{enumerate}[label=(\arabic*)]
        \item There exists an algorithm that, for any $k \in \N$, constructs an $O(c k)$-distortion spanner with $\tilde{O}(kn^{1+1/2k} \log \Delta)$ edges and runs in time $\tilde{O}(\bnntime(n) \cdot k n^{1/k} \log \Delta)$.
        \item Any algorithm that learns at most $o(n^{1+1/k}/k)$ pairwise distances by querying a distance oracle and a blackbox batch nearest neighbor search oracle necessarily incurs $\Omega(c k)$ distortion.
    \end{enumerate}
    
Our results entail that (truly) sub-quadratic time algorithms for spanner construction is {\em equivalent} to subquadratic time BANN (up to constant-factor losses).

    As a further application, we use our fast spanner constructions to obtain a fast algorithm for approximating the Wasserstein distance $\sfW_q$, for all $q > 1$, over any metric space admitting an efficient batch approximate nearest neighbor search algorithm. Together with recent new efficient algorithms for approximate nearest neighbor search in $\ell_p$ spaces, for $p > 2$, our results entail the first subquadratic time algorithms for spanner construction (with the stated size-distortion tradeoff) and $\sfW_q$ distance approximation over these metric spaces. 
\end{abstract}
\newpage
\newcommand{\LabA}{ANN}
\newcommand{\LabB}{BANN}
\newcommand{\LabC}{Spanner}
\newcommand{\LabD}{EMD}
\newcommand{\LabE}{$\sfW_q$}
\newcommand{\LabF}{MST}

\newcommand{\ABabove}{}   \newcommand{\ABbelow}{}   
\newcommand{\ADabove}{\cite{AS14}}   \newcommand{\ADbelow}{}   
\newcommand{\BCabove}{?}   \newcommand{\BCbelow}{}   
\newcommand{\CBabove}{}   \newcommand{\CBbelow}{}   
\newcommand{\BDabove}{\cite{BCJW25}}   \newcommand{\BDbelow}{}   
\newcommand{\AEabove}{\cite{AS12}}   \newcommand{\AEbelow}{}   
\newcommand{\AFabove}{\cite{HIM12}}   \newcommand{\AFbelow}{}   
\newcommand{\CDabove}{\cite{CKLPPS22}}   \newcommand{\CDbelow}{}   
\newcommand{\CEabove}{\cite{CKLPPS22}}   \newcommand{\CEbelow}{}   
\newcommand{\CFabove}{}   \newcommand{\CFbelow}{}   

\newcommand{\gapAB}{2.5}      
\newcommand{\gapBC}{3}      
\newcommand{\gapCcol}{3}    

\newcommand{\gapDE}{1.5}      
\newcommand{\gapEF}{1.5}      
\newcommand{\stackShift}{0}   

\newcommand{\bdOut}{72}   \newcommand{\bdIn}{170}   
\newcommand{\aeOut}{-50}   \newcommand{\aeIn}{190}   
\newcommand{\afOut}{-72}  \newcommand{\afIn}{190}   
\newcommand{\adOut}{80}   \newcommand{\adIn}{150}   

\newcommand{\posAB}{0.50}
\newcommand{\posCB}{0.50}
\newcommand{\posAD}{0.50}
\newcommand{\posBC}{0.50}
\newcommand{\posBD}{0.60}
\newcommand{\posAE}{0.5}
\newcommand{\posAF}{0.60}
\newcommand{\posCD}{0.50}
\newcommand{\posCE}{0.50}
\newcommand{\posCF}{0.62}

\newcommand{\xB}{\gapAB}
\newcommand{\xC}{\gapAB+\gapBC}
\newcommand{\xCol}{\gapAB+\gapBC+\gapCcol}

\section{Introduction}

A graph $t$-spanner for a graph $G$ on $n$ vertices with $m$ edges is a subgraph $H \subseteq G$ for which the shortest path distances satisfy
\begin{align*}
    \forall i, j \in [n]. \, \, \sfd_G(i, j) \leq \sfd_H(i, j) \leq t \cdot \sfd_G(i, j)
\end{align*}
where $\sfd_G$ and $\sfd_H$ denote the shortest path metrics in $G$ and $H$ respectively. The notion was first introduced by \cite{PS89}, and has had wide-reaching influence\footnote{We refer the interested reader to \cite{ABS+20, Epp96, Zwi01} for comprehensive surveys on the many developments and uses of graph spanners.}. In the design of spanners, one is traditionally concerned with the \emph{size-distortion} trade-off. That is, the trade-off between the size of subgraph $H$ and the incurred distortion $t$. A classic result here is the existence of $(2k-1)$-spanners of size $O(n^{1 + 1/k})$ for $k \in \N$ and in any metric \cite{PS89}. Such a trade-off is known to be tight under the Erdös Girth Conjecture \cite{Erdos63}. Moreover, spanners achieving this trade-off can be constructed via a simple greedy algorithm with runtime $O(m)$ \cite{PS89}. 

In the \emph{geometric} setting \cite{chew1986planar}, one is tasked with constructing a spanner for the complete graph on $n$ points in a metric space $\calM = ([n], \sfd_{\calM})$ (where pairwise edge weights are given by the metric distance). Here, $m = O(n^2)$ but one could hope for subquadratic-in-$n$ time algorithms since the input description is only linear in $n$. Counter-intuitively, it is not even clear whether one can build the aforementioned classic spanner in sub-quadratic time. The only metric where such progress has been made is the Hamming and Euclidean spaces, where there has been a long line of work on the construction of spanners for both low-dimensional \cite{Clark87, Keil88, KG92, ADDJS93, Giri07} and high-dimensional \cite{HIS13, AZ23, JPSWZ26} spaces, with varying size-distortion trade-offs. This line of work naturally raises the question for other metrics:
\begin{question}
\label{q:fastSpanner}
    {\em For which metric spaces do  there exist subquadratic time algorithms to construct geometric spanners with non-trivial distortion?}
\end{question}

The existing sub-quadratic time spanner constructions typically leverage efficient algorithms for Approximate Nearest Neighbor Search (ANN), or the tools behind the ANN algorithms (such as Locality Sensitive Hashing \cite{HIS13}). Indeed, ANN seems to be a related problem for two reasons, at least in the case of offline version of ANN  termed Batch ANN (BANN), where one has to solve ANN for a query set given in advance~\cite{alman2015probabilistic}. (1) Some down-stream applications of spanners can be solved via efficient BANN (see related work in Section \ref{sec: related-work}). (2) One can efficiently solve BANN for a pointset given a spanner on that pointset in a black-box fashion. See Figure~\ref{fig: landscape} for a depiction of connections between some central problems. These connections raise the question of how tightly BANN and spanner construction problems are related, namely whether (2) holds in the reverse direction.

\begin{question}
\label{q:spannerBANN}
 {\em    
 Does there exist a black-box reduction from geometric spanner construction to Batch Approximate Nearest Neighbor Search (BANN), under an arbitrary metric?}
\end{question}

In this work, we make significant progress on both of these questions. In particular, we show that given black-box access to a BANN algorithm in an arbitrary metric space $\calM = ([n], \sfd)$, there exists an algorithm that makes sublinear-in-$n$ many calls to the oracle and constructs spanners with natural size-distortion trade-offs. Overall, when the runtime of the oracle is (truly) subquadratic-in-$n$, the spanner construction takes subquadratic time as well.  

A consequence of our results is that, for a fixed approximation $t>1$, existence of a (truly) sub-quadratic time algorithm for $t$-spanner construction is {\em equivalent} to subquadratic time BANN (up to constant-factor losses). 

A classic motivation for building geometric spanners is that it allows one to lift the rich literature on graph algorithms to geometric problems. Indeed, we discuss one concrete such application of our results to fast algorithms for Wasserstein-$q$ distance\footnote{We refer the reader to \cite{CP20, P26, PZ19} for surveys on the practical usefulness of $\sfW_q$ distance in applications such as statistics, data analysis, and machine learning.} approximation, including Earth Mover Distance (EMD), over metric spaces with efficient batch approximate nearest neighbor search, such as $\ell_p$ with $p\in (2,\infty]$. 


\begin{figure}[t] 
    \centering
    \begin{tikzpicture}[
        >={Stealth[length=2.8mm,width=2.1mm]},
        vertex/.style   = {black, inner sep=2pt, minimum size=7mm},
        flow/.style     = {->, black, thick},
        dotflow/.style  = {flow, densely dotted},
        edgecap/.style  = {black, font=\small, inner sep=2.5pt},
    ]

    \node[vertex] (A) at (0,0)                              {\LabA};
    \node[vertex] (B) at ({\xB},0)                          {\LabB};
    \node[vertex] (C) at ({\xC},0)                          {\LabC};
    \node[vertex] (D) at ({\xCol},{\stackShift+\gapDE})     {\LabD};
    \node[vertex] (E) at ({\xCol},{\stackShift})            {\LabE};
    \node[vertex] (F) at ({\xCol},{\stackShift-\gapEF})     {\LabF};
    
    
    \draw[flow] (A) --
        node[edgecap, above, sloped, pos=\posAB] {\ABabove}
        node[edgecap, below, sloped, pos=\posAB] {\ABbelow} (B);
    
    \draw[dotflow, transform canvas={yshift = 3pt}] (B) --
        node[edgecap, above, sloped, pos=\posBC] {\BCabove}
        node[edgecap, below, sloped, pos=\posBC] {\BCbelow} (C);

    \draw[flow, transform canvas={yshift =-3pt}] (C) --
        node[edgecap, above, sloped, pos=\posBC] {\CBabove}
        node[edgecap, below, sloped, pos=\posBC] {\CBbelow} (B);
    \draw[flow] (B) to[out=\bdOut, in=\bdIn]
        node[edgecap, above, sloped, pos=\posBD] {\BDabove}
        node[edgecap, below, sloped, pos=\posBD] {\BDbelow} (D);
    
    \draw[flow] (A) to[out=\aeOut, in=\aeIn]
        node[edgecap, above, sloped, pos=\posAE] {\AEabove}
        node[edgecap, below, sloped, pos=\posAE] {\AEbelow} (E);
    
    \draw[flow] (A) to[out=\afOut, in=\afIn]
        node[edgecap, above, sloped, pos=\posAF] {\AFabove}
        node[edgecap, below, sloped, pos=\posAF] {\AFbelow} (F);
    
    \draw[flow] (C) --
        node[edgecap, above, sloped, pos=\posCD] {\CDabove}
        node[edgecap, below, sloped, pos=\posCD] {\CDbelow} (D);
    
    \draw[flow] (C) --
        node[edgecap, above, sloped, pos=\posCE] {\CEabove}
        node[edgecap, below, sloped, pos=\posCE] {\CEbelow} (E);
    
    \draw[flow] (C) --
        node[edgecap, above, sloped, pos=\posCF] {\CFabove}
        node[edgecap, below, sloped, pos=\posCF] {\CFbelow} (F);

    \draw[flow] (A) to[out=\adOut, in=\adIn]
        node[edgecap, above, sloped, pos=\posAD] {\ADabove}
        node[edgecap, below, sloped, pos=\posAD] {\ADbelow} (D);

    \end{tikzpicture}
    \caption{Depiction of the computational landscape, featuring some important geometric problems that are known to have connections to ANN/BANN and spanners. $A\to B$ means having efficient algorithm for $A$ implies efficient algorithm for $B$ (i.e., $B$ reduces to $A$) with approximation that is preserved up to $O(\cdot)$. The reduction from Earth Mover's Distance (EMD) to ANN in \cite{AS14} holds for general metric spaces, whereas the reduction from EMD to BANN in \cite{BCJW25} holds specifically in $\ell_p$ spaces where $p \in [1, 2]$. The reduction from $\sfW_q$ to ANN in \cite{AS12} holds for $q = 1, 2$ and requires the ANN data structure to be dynamic. The reduction from ANN to BANN is the trivial one. The spanner$\to W_q$/EMD connection is via min-cost flow algorithm of \cite{CKLPPS22}. The dotted arrow captures our conceptual question and is the focus of this paper.}
    \label{fig: landscape}
\end{figure}
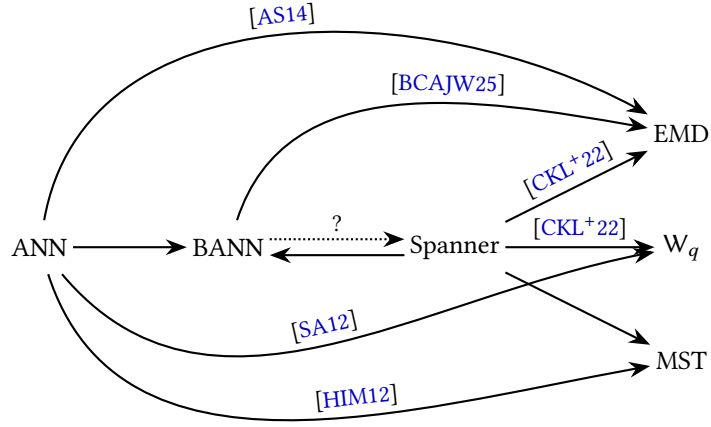



\subsection{Our Results}

We now state more formally our results: new algorithms that construct geometric spanners given access to a Batch Approximate Nearest Neighbor oracle in \emph{general} metric spaces, as well as a lower bound quantifying the best possible trade-offs. In particular, we define the oracle model as follows. 

\begin{definition}[Batch Approximate Nearest Neighbor Search Oracle] \label{def: bann}
    Let $\calM = ([n], \sfd)$ denote an arbitrary metric, let $c \geq 1$. We say that $\calM$ admits an efficient batch $c$-approximate nearest neighbor ($c$-BANN) search oracle $\calO$ for approximation factor $c \geq 1$ if the following is true: for any query set $Q \subseteq [n]$ and target set $X \subseteq [n]$, $\calO(Q, X, c)$ and with probability $\geq 1 - 1/n$ returns $\{ (q, x^{\ast}, \sfd(q, x^{\ast})) \}_{q \in Q}$ where $\sfd(q, x^{\ast}) \leq c \cdot \min_{x \in X} \sfd(q, x)$ for each $q \in Q$. We let $\bnntime(n)$ denote the runtime of $\calO$.  
\end{definition}

We remark that there exist many metric spaces which admit efficient Batch Approximate Nearest Neighbor Search (BANN) data structures with non-trivial approximations. Examples include $\ell_p$ spaces for $p \in [1, 2]$ \cite{AR15, ACW16, le2014algorithms} and for $p > 2$ \cite{Kush2021Near,AN25}, as well as Schatten-$1$/Schatten-$2$ \cite{Kush2021Near}, symmetric normed spaces \cite{Andoni2017Approximate}, and general norm spaces~\cite{Andoni2018Holder}.

\subsubsection{Spanner Construction and its Tightness} Our main algorithmic result is the following spanner construction via blackbox access to a BANN oracle.

\begin{theorem}\label{thm: upper-bound}
    Let $\calM = ([n], \sfd)$ be a metric, with aspect ratio $\Delta$\footnote{Here, and throughout the paper, the aspect ratio $\Delta := \max_{i, j \in [n]} \sfd(i, j) / \min_{i \neq j \in [n]} \sfd(i, j)$. That is, the ratio between the maximum and minimum pairwise distances in the finite metric}, that admits a $c$-BANN oracle $\calO$ with runtime $\bnntime$, where $c \geq 1$. Then, there exists an algorithm that constructs, for any $k \in \N$, an $O(ck)$-distortion spanner with $\Tilde{O}(kn^{1+1/2k} \log \Delta )$ edges in time $\Tilde{O}( \bnntime(n) \cdot kn^{1/k} \log \Delta)$ and succeeds with high probability. Moreover, the resulting spanner is an $O(k)$-hop spanner.   
\end{theorem}

Theorem \ref{thm: upper-bound} implies the following strong equivalence between BANN and spanner construction in general metrics. 

\begin{corollary}
    Let $\calM = ([n], \sfd)$ be a metric, with aspect ratio $\Delta$. Suppose there exists an algorithm for $c$-BANN in $\calM$ with runtime $n^{2-\delta}$ for $\delta>0$. Then there exists an algorithm that constructs an $O_\delta(c)$-distortion spanner for $\calM$ in time $O(n^{2-\Omega(\delta)}\log \Delta)$. Similarly, if there exists an algorithm to construct a $t$-spanner of $\calM$ with runtime $n^{2-\delta}$, then there exists an algorithm for $t$-BANN in $\calM$ with runtime $\tilde{O}(n^{2-\delta})$.
\end{corollary}

We also prove a corresponding lower bound to Theorem \ref{thm: upper-bound}, showing that distortion $\Omega(ck)$ is required if performing $n^{1/k}$ BANN queries (even with $n^{1+1/k}$ extra exact distance computation calls).

\begin{theorem}[informal, see Theorem \ref{thm: ck-lower-bound}] \label{thm: lower-bound}
    Let $c > 1$ and let $k \geq 2$. Then for every randomized algorithm that learns at most
    \[
        \Phi \;=\; o\!\left( \frac{n^{1+1/k}}{k} \right)
    \]
    pairwise distances by querying a $c$-BANN oracle and an exact distance oracle, there is a metric space on which the algorithm incurs distortion 
    $t \;=\; \Omega(ck)$
    in expectation over the algorithm's own randomness. 
    
\end{theorem}





\subsubsection{Applications}
As an application of Theorem \ref{thm: upper-bound}, we show the existence of efficient algorithms for estimating the Wasserstein distance $\sfW_q$ over general metrics. The $\sfW_q$ distance is the optimal transport distance; for $q = 1$, $W_1$ is also known as Earth Mover's Distance (EMD). The $\sfW_q$ distance is defined as follows: 

\begin{definition}[Wasserstein-$q$ Distance] \label{def: wq}
    Fix a metric space $\calM = ([n], \sfd)$ and let $q \in [1, \infty)$. Let $\mu, \nu$ denote probability distributions supported on $[n]$. The Wasserstein-$q$ distance over $\calM$ is defined as
    \begin{align*}
        \sfW_{q}(\mu, \nu) := \min_{\gamma \in \Gamma(\mu, \nu) }  \left( \Ex_{(\bx, \by) \sim \gamma} \left[\sfd(\bx, \by)^q \right] \right)^{1/q}
    \end{align*}
    where $\Gamma(\mu, \nu)$ denotes the set of couplings\footnote{That is, $\Gamma(\mu, \nu) :=  \{ \gamma \in \R^{n \times n} : \sum_{i \in [n]} \gamma(i, j) = \nu(j) \text{ and } \sum_{j \in [n]} \gamma(i, j) = \mu(i)\}$. Additionally, $\supp(\gamma) := \{(x, y) \in [n] \times [n]: \gamma(x, y) > 0 \}$.} of $\mu$ and $\nu$. When $q = \infty$, we define
    \begin{align*}
        \sfW_{\infty} := \min_{\gamma \in \Gamma(\mu, \nu)} \max_{(x, y) \in \supp(\gamma)} \sfd(x, y)
    \end{align*}
\end{definition}

To estimate the $W_q$ distance, we can combine our main theorem together with a standard reduction to the problem of Minimum Cost Flow (MCF), for which we use the breakthrough algorithm of \cite{CKLPPS22} solving MCF in near-linear time. 

\begin{corollary}[$W_q$ Application] \label{cor: wasserstein-application}
    Let $\calM = ([n], \sfd)$ be a metric with aspect ratio $\Delta$, that admits a $c$-BANN oracle $\calO$ with runtime $\bnntime$, where $c \geq 1$. Then, for arbitrary discrete probability distributions $\mu, \nu \in \frac{1}{U} \Z^n$, where $U \in \N$, there exists an algorithm that computes an $O(c k^{2- 1/q})$ approximation to $\sfW_q$ over $\calM$ for $q \in [1, \infty]$. The runtime of the algorithm is $(k^2 n^{1+1/2k})^{1 + o(1)} \polylog (U \Delta) + \Tilde{O}(\bnntime(n) \cdot k n^{1/k} \log \Delta)$ and the algorithm succeeds with high probability. 
 \end{corollary}


As a concrete instance of Theorem ~\ref{thm: upper-bound} and Corollary \ref{cor: wasserstein-application}, we obtain the best known results in the setting of high-dimensional $\ell_p$ space when $p > 2$. To do so, we use the ANN algorithm from \cite{AN25} to obtain the best known (i) time-size-distortion geometric spanner tradeoffs and (ii) time-approximation tradeoffs for approximate $\sfW_q$ over $\ell_p$ when $p > 2$. 

\begin{theorem}[Theorem 1.7 of \cite{AN25}]
    For any $d \geq 1$ and $\epsilon > 0$, $p > 1$, there exists an ANNS data structure under $X \subset (\R^d, \ell_p)$ where $|X| = n$ achieving $O(\log p / \epsilon)$ approximation, $n^{\epsilon} d^{O(1)}$ query time, and $n^{1+\epsilon} d^{O(1)}$ space and preprocessing time.  
\end{theorem}

\begin{corollary}[$\ell_p$ Instantiation] \label{cor: ellp-instantiation}
    Let $d \geq 1, \epsilon > 0, p > 1$ and $X \subset (\R^d, \ell_p)$ with aspect ratio $\Delta$ and $|X| = n$. Then, for any $k \in \N$,
    \begin{enumerate}[label=(\roman*)]
        \item There exists an algorithm that constructs an $O(k \log p /\epsilon)$-distortion (undirected) spanner with $\Tilde{O}(kn^{1+1/2k} \log \Delta )$ edges in time $n^{1+\epsilon} d^{O(1)} \cdot kn^{1/k} \log \Delta$ and succeeds with high probability. Moreover, the spanner is an $O(k)$-hop spanner. 
        \item For arbitrary discrete probability distributions $\mu, \nu \in \frac{1}{U} \Z^n$, where $U \in \N$, there exists an algorithm that computes an $O(k^{2- 1/q} \log p/ \epsilon)$ approximation to $\sfW_q$ over $\calM$ for $q \in [1, \infty]$. The runtime of the algorithm is $(k^2 n^{1+1/2k})^{1 + o(1)} \polylog (U\Delta)+ n^{1+\epsilon} d^{O(1)} \cdot k n^{1/k} \log \Delta$ and the algorithm succeeds with high probability.
    \end{enumerate} 
    
\end{corollary}

Similarly, best known algorithms can be obtained for other spaces with specialized ANN algorithms, including Schatten-$1$/Schatten-$2$ \cite{Kush2021Near}, symmetric normed spaces \cite{Andoni2017Approximate}, and general norm spaces~\cite{Andoni2018Holder}.\footnote{Some of these algorithms have polynomial-in-$n$ preprocessing times (rather than near-linear), but a standard reduction can reduce preprocessing time to near-linear, at the cost of increasing query time (still sublinear).}


\subsection{Related Work} \label{sec: related-work}

\paragraph{Geometric spanners.}
There is a long line of work on geometric spanner constructions, beginning with classical constructions in low-dimensional Euclidean space \cite{Chew86, Keil88,KG92} and extending more recently to high-dimensional settings \cite{HIS13,AZ23,JPSWZ26}. In particular, \cite{AZ23} gives subquadratic directed Steiner spanners in high-dimensional Euclidean space, while \cite{JPSWZ26} gives nearly optimal directed Steiner spanners for high-dimensional $\ell_p$ spaces with $p\in[1,2]$. For $p>2$, however, there remains a gap between existential and subquadratic-time constructions. Recent work \cite{KP25} shows the existence of spanners with $O(t)$ distortion and $\widetilde O(n^{1+1/t^q})$ size, where $q=\frac{p}{p-1}$, improving upon the size-distortion tradeoff of the standard greedy construction \cite{ADDJS93}, but does not provide a subquadratic construction-time bound. Our work instead asks what spanner guarantees can be obtained in general metric spaces from black-box access to efficient Batch Approximate Nearest Neighbor Search.

\paragraph{Approximate nearest neighbor search as a primitive.}
Approximate nearest neighbor search and techniques such as locality-sensitive hashing have been extensively studied in high-dimensional spaces \cite{IM98,DIIM04,AI06,AR15,AR16}. Beyond nearest neighbor search itself, these techniques have also been used as algorithmic tools for approximate minimum spanning trees~\cite{HIM12}, hierarchical clustering~\cite{MVW21}, and geometric spanner constructions \cite{HIS13,AZ23,JPSWZ26}. 


\paragraph{Wasserstein-$q$ Approximation in General Metric Spaces.}

There has been much prior work on efficient algorithms for optimal transport under general and geometric cost functions. A large body of work studies additive approximations and entropically regularized optimal transport \cite{Cut13,AWR17,DGK18,LHJ19}, while fewer results are known for relative approximations to Wasserstein-$q$ distance over arbitrary metric spaces. 

The recent work \cite{LRSS25} gives, for every constant $q\ge2$, an $O(\log n)$-approximation to $W_q$ in $O(n^2\log U\log\Delta\log n)$ time, while \cite{ACRY26} gives a $(4+\epsilon)$-approximation for all $q\in[1,\infty]$. For $q=1$ over uniform distributions, \cite{AS14} gives algorithms for metric bipartite matching using dynamic approximate nearest neighbor search, and there is also a literature on embeddings, sketches, and streaming algorithms for EMD \cite{NS07, ADBIW09}. In contrast, our reduction uses static offline BANN and first constructs a bounded-hop spanner before reducing to minimum-cost flow.

\paragraph{Wasserstein-$q$ Approximation in Low-Dimensional $\ell_p$ Spaces.}

Faster algorithms are known in low-dimensional Euclidean spaces. For $q=1$, several works give subquadratic and near-linear-time approximation schemes for geometric transportation \cite{AFPVX17,KNP19, FL19}. For $q=2$ over uniform distributions, \cite{AP06} gives a $(1+\epsilon)$-approximation in $\widetilde O((n/\epsilon)^{3/2})$ time for planar point sets, while \cite{LR21} improves this to $\widetilde O(n^{5/4}\operatorname{poly}(1/\epsilon))$. Related work also gives subquadratic algorithms for $W_\infty$ over uniform distributions through bottleneck matching \cite{LR21}.

\paragraph{Wasserstein-$q$ Approximation in High-Dimensional $\ell_p$ Spaces.}

Prior work has also studied EMD through embeddings, sketching, and nearest-neighbor search, including in high-dimensional settings \cite{IT03,KN06,AIK08,BDIRW20,CJLW22,JWZ24}. In particular, \cite{AIK08} gives an $O(\log s\log d)$-distortion embedding for sets of cardinality at most $s$ in dimension $d$, together with corresponding lower bounds for sketching and communication. More recently, geometric spanner constructions have led to truly subquadratic algorithms for EMD and Wasserstein distances. The high-dimensional Euclidean construction of \cite{HIS13} gives an $O(c)$-distortion spanner with $\widetilde O(n^{1+1/c^2})$ edges, which can be combined with minimum-cost flow to obtain an $O(c)$ approximation for EMD. The work of \cite{AZ23} subsequently gives $(1+\epsilon)$-approximate directed Steiner spanners of size $n^{2-\Omega(\epsilon^2)}$ and a $(1+\epsilon)$-approximation for high-dimensional Euclidean EMD in $n^{2-\Omega(\epsilon^2)}$ time. Recent work of \cite{JPSWZ26} proves nearly optimal $(1+\epsilon)$-approximate directed Steiner spanners for high-dimensional $\ell_p$, for $p\in[1,2]$, with $\widetilde O(n^{2-\Omega(\epsilon)})$ edges and nearly matching construction time. As an application, they provide faster algorithms for Wasserstein-$q$ distances for $q\in[1,\infty)$. 

For $q=1$, \cite{BCJW25} instead bypasses spanner construction by reducing $(1+\epsilon)$-approximate EMD in high-dimensional $\ell_p$, for $p\in[1,2]$, to $(1+\epsilon)$-approximate bichromatic closest pair. Combined with~\cite{ACW16}, this gives a $(1+\epsilon)$-approximation in $n^{2-\widetilde{\Omega}(\epsilon^{1/3})}$ time, improving upon the previous $n^{2-\Omega(\epsilon^2)}$ bound of \cite{AZ23}.

\subsection{Technical Overview} \label{sec: technical-overview}

We now provide an overview of the relevant ideas and conceptual contributions made by this work. We first discuss our upper bound (Theorem \ref{thm: upper-bound}), then the lower bound construction (Theorem \ref{thm: lower-bound}), and then finally the application to Wasserstein distance approximation (Corollary \ref{cor: wasserstein-application}). 

\subsubsection{Upper Bound}

We start with the goal of answering Question \ref{q:spannerBANN}. Our first observation is that any such algorithm \emph{must} recover efficient algorithms for graph spanner constructions. This is due to the fact that any graph $G$ with $n$ vertices and $m$ edges trivially admits a BANN algorithm with runtime $\tilde O(m)$ that simply runs multi-source shortest paths from the query set. Thus, a natural starting point is to adapt existing spanner or distance oracle constructions for graph metrics. 

A natural candidate algorithm is the  classic Thorup-Zwick distance oracle construction \cite{TZ05}, especially due to its use of ``nearest neighbor type'' operations. Introducing an idea, described below, we can obtain a form of Theorem~\ref{thm: upper-bound}, but the overall approximation comes out to be $\approx c^k$.


To obtain a better (optimal) dependence on $k$ of Theorem \ref{thm: upper-bound}, we develop 
a geometric version of the Baswana-Sen graph spanner construction \cite{BS07} instead. Next, we describe the steps taken to reach that construction and the intuition that led there. 


\paragraph{A First Attempt: Thorup-Zwick \cite{TZ05}.} The Thorup-Zwick distance oracle construction is a natural starting point since it is built almost entirely out of nearest neighbor queries. One samples a nested sequence $[n] = A_0 \supseteq A_1 \supseteq \dots \supseteq A_{k-1}$, keeping each point of $A_i$ in $A_{i+1}$ independently with probability $n^{-1/k}$. We let $p_i(v)$ denote the point of $A_i$ nearest to $v$, and define the \emph{bunch} of $v$ to be 
\begin{align*}
    B(v) := \bigcup_{i} \left \{ u \in A_i \mid \sfd(v, u) < \sfd(v, p_{i+1}(v)) \right\},
\end{align*}
the points of each level that are strictly closer to $v$ than $v$'s pivot at the next level. The spanner consists of the edges from every $v$ to its pivots and to every vertex in its bunch. Note that the bunches have small size in expectation. To see this, consider listing all $n$ points in order of increasing distance from $v$. Each point is promoted to $A_{i+1}$ independently with probability $n^{-1/k}$, so one expects to encounter $n^{1/k}$ points before the first promoted one, giving $|B(v)| = O(kn^{1/k})$ in expectation. In the original construction, the pivots and bunches are found via brute force distance computations which ultimately contribute to the quadratic runtime.

The distortion guarantee of the implicitly stored spanner comes from an alternating search. To connect $u$ and $v$ at distance $\sfd(u,v)$, one checks whether $p_i(u) \in B(v)$. If so, the path $u \to p_i(u) \to v$ lies in the spanner and we stop, and if not the roles of $u$ and $v$ are swapped and $i$ is incremented. Observe that $A_{k-1} \subseteq B(w)$ for every vertex $w$, so the process must terminate after at most $k-1$ steps. The point is that failure of the check is informative. By definition of the bunch, $p_i(u) \notin B(v)$ means $\sfd(v, p_i(u)) \geq \sfd(v, p_{i+1}(v))$, so 
\begin{align*}
    \sfd(v, p_{i+1}(v)) \, \leq \, \sfd(v, p_i(u)) \, \leq \, \sfd(u,v) + \sfd(u, p_i(u))
\end{align*}
Thus, $\sfd(v, p_{k-1}(v)) \leq (k-1) \cdot \sfd(u, v)$ by a telescoping sum, and we can conclude that the spanner path $u \to p_{k-1}(v) \to v$ has length $\sfd(u, p_{k-1}(v)) + \sfd(v, p_{k-1}(v)) \leq (2k-1 ) \cdot \sfd(u, v)$. Observe that the size of the implicitly stored spanner is simply $\widetilde{O}(n^{1+1/k})$.

Now, consider an implementation of the Thorup-Zwick distance oracle via $c$-BANN oracle calls. Any such algorithm must compute suitable approximate versions of the pivots and bunches via $c$-BANN. Finding approximate pivots amounts to a simple BANN call on target sets $A_0, \ldots, A_{k-1}$, however, it is not immediately clear how to compute the vertex bunches $B(v)$ quickly. 

A $c$-BANN oracle does not return $p_{i+1}(v)$, but only some $\tilde{p}_{i+1}(v)$ with $\sfd(v, \tilde{p}_{i+1}(v)) \leq c \cdot \sfd(v, A_{i+1})$, and bunches are accordingly defined by thresholding against this approximate distance. There are two natural ways to consider defining this threshold
\begin{enumerate}
    \item Thresholding at $\sfd(v, \tilde{p}_{i+1}(v))$ lets a bunch extend a factor $c$ past the first sampled point. In this case, however, we can no longer guarantee that the bunches have bounded size (and therefore, that the spanner size is bounded).
    \item Thresholding instead at $\sfd(v, \tilde{p}_{i+1}(v))/c \leq \sfd(v, A_{i+1})$ retains the guarantee of bounded bunch sizes. In this case, we can compute bunches efficiently, using a ``hashing trick'' that we will discuss below, in The Crux paragraph in the discussion of the Baswana-Sen construction. However, the Thorup-Zwick construction necessarily incurs exponential-in-$k$ distortion. This is because a failed check only certifies $\sfd(v, \tilde{p}_{i+1}(v)) \leq c \cdot \sfd(v, p_i(u))$, and the recursion above degrades to
    \begin{align*}
        \sfd(v, \tilde{p}_{i+1}(v)) \, \leq \, c (\sfd(u,v) + \sfd(u, p_i(u)))
    \end{align*}
    which unrolls to give a distortion exponential in $k$. The difficulty comes from the chained levels that force the factor of $c$ to be paid $k$ times. 
\end{enumerate}

\paragraph{A Second Attempt: Baswana-Sen Spanners \cite{BS07}.} The Baswana-Sen construction allows us to charge the factor $c$ cost additively rather than multiplicatively $k$ times. The original construction \cite{BS07} for unweighted graphs maintains a partition of a shrinking set of \emph{active} vertices into \emph{clusters}, each with a designated center. Initially every vertex forms its own cluster, then in each of $k-1$ iterations, every cluster is sampled independently with probability $n^{-1/k}$. 

An active vertex $v$ lying in an unsampled cluster does one of two things. If $v$ has a neighbor in some sampled cluster, it adds one edge to the nearest such neighbor and joins that cluster. Otherwise $v$ is \emph{retired}, and adds one edge to each of its adjacent clusters. Observe that at iteration $i$, every cluster has radius at most $i$, since it grows by at most one per iteration. Furthermore, a retired vertex has $O(n^{1/k})$ adjacent clusters in expectation, since each cluster was sampled independently with probability $n^{-1/k}$ and none of $v$'s adjacent clusters was sampled. This gives $O(k)$ distortion with size $O(n^{1+1/k})$ in expectation. 

Now, we consider the changes needed to run this in a metric space with a BANN oracle. First, in a metric space all pairs are adjacent, so adjacency must be replaced by a distance scale. We fix $r = 2^{\ell}$ for $\ell \in [\log \Delta]$, run the entire construction independently at each scale, and return the union of the resulting edge sets. A pair $x,y$ is then handled at the scale with $r/2 < \sfd(x,y) \leq r$. Second, the step ``join the nearest sampled cluster'' is precisely a batch nearest neighbor query. At iteration $i$ of scale $r$, letting $S_{\ell, i}$ denote the vertices lying in sampled clusters and $U_{\ell,i}$ the active vertices lying in unsampled clusters, the single call $\calO(U_{\ell, i}, S_{\ell, i}, c)$ serves every vertex at once. If the returned point $u$ satisfies $\sfd(v,u) \leq c^2 r$ then $v$ adds the edge $(v,u)$ and joins $u$'s cluster, and otherwise $v$ is retired. A retired vertex certifies that no cluster within distance $cr$ of it was sampled. This is because any such close cluster being sampled would then surface a close vertex $u$, preventing $v$ from being retired in the first place. Exactly as in the original construction, a retired vertex therefore has at most $\Tilde{O}(n^{1/k})$ clusters within distance $cr$ of it, with high probability.

\paragraph{The Crux: the ``hashing trick''.} Two difficulties stand in the way. First, when a vertex is retired, it is not feasible to find and add a suitable short edge to each nearby cluster. An analogous step to that in the original construction would entail one oracle call per cluster, of which there are potentially very many. On the other hand, a single call with target set all of $U_{\ell, i}$ returns only one point, and hence only one edge. Second, the call $\calO(U_{\ell, i}, S_{\ell, i}, c)$ is made to play two roles at once. (1) It must supply the attachment edge of $v$, and (2) it must certify that no sampled cluster lies near $v$. Because the oracle answers only up to a factor of $c$, the only way to conclude that no sampled cluster lies within distance $cr$ of $v$ (the radius out to which the retirement step must reach) is to retire $v$ when the returned distance exceeds $c^2 r$. Every shorter answer must then be accepted as an attachment edge, so cluster radii grow by $c^2 r$ per iteration and the construction described so far incurs distortion $O(c^2 k)$.

Both difficulties are resolved by the same device, \emph{isolating} clusters by hashing. At iteration $i$, and prior to sampling, every active vertex $v$ assembles a \emph{candidate list} $L_v$ of at most $\tau$ pairs $(C, z)$ with $z \in C$ and $\sfd(v, z) \leq cr$. To build these lists, we hash the active clusters $\calC_{\ell,i}$ into $B = \Theta(\tau)$ buckets under a uniformly random hash function, and call the oracle once per bucket $j$. Here, the target set is the active vertices whose cluster hashes to $j$, and the query set consists of those $v$ whose list is not yet full and none of whose \emph{known} clusters hashes to $j$. Any answer of length at most $cr$ is appended to $L_v$, and the restriction on the query set guarantees that it belongs to a cluster $v$ has not yet recorded, so no oracle call is wasted. A cluster within distance $r$ of $v$ that is still missing from $L_v$ avoids all of the at most $\tau + 1$ known clusters with constant probability, in which case its bucket is queried and the oracle must answer with a point at distance at most $cr$. Each repetition of the hashing therefore adds a new cluster to $L_v$ with constant probability, and $\Tilde{\Theta}(\tau)$ repetitions suffice for every list to either fill up or contain every cluster within distance $r$ (Lemma \ref{lem: list-completeness}).

Only now do we sample, and $v$ attaches to a sampled cluster of $L_v$ if one exists, retiring otherwise. Because attachment is drawn from within a list that was fixed beforehand, the attachment edge has already been certified at threshold $cr$ and no second approximation is incurred. That is, cluster radii grow by $cr$ per iteration rather than $c^2 r$. The dichotomy above is precisely what the two cases require. If $L_v$ is full, it holds $\tau$ distinct clusters, each sampled independently, so with high probability one of them is sampled and $v$ attaches (Lemma \ref{lem: full-lists-sampled}). If instead $v$ retires, then $L_v$ was not full and is therefore complete, so the edges added at retirement reach every cluster within distance $r$ of $v$ (Lemma \ref{lem: hashing-correctness}). 




\subsubsection{Lower Bound}

We now describe the ideas behind Theorem \ref{thm: lower-bound}, that any reduction from spanner construction to BANN must incur $\Omega(ck)$ distortion, if the algorithm only learns  $o(n^{1+1/k}/k)$ pairs via the BANN oracle and an exact pairwise distance oracle. We note that the upper bound does not use any (exact) distance oracle calls, but we find it natural to allow the reduction to use a proportionate number of distance oracle calls.

To establish the lower bound we exhibit an explicit hard metric, as well as adversarial answers to the BANN calls.
Our hard instance needs to simultaneously incorporate two sources of approximation:
\begin{enumerate}
    \item The $\Omega(c)$ distortion: from the $c$-BANN oracle hiding the true nearest neighbor. 
    \item The $\Omega(k)$ distortion: from the reduction using too few oracle calls.
\end{enumerate}

  The main challenge is to combine these two sources of approximation such that their distortion composes {\em multiplicatively} rather than additively. 
  
  The hard instance is as follows. First of all, note that any added spanner edge weight should always be equal to the distance between its endpoints in the original metric (smaller weight would violate the approximation guarantee, and a larger weight provides no advantages). Therefore we can assume the algorithm only adds edges between pairs that are queried against the oracles, and its spanner has at most $n^{1+1/k}$ edges (this is how many distances we learn from the allowed number of oracle calls). Now, let $G$ be a graph of girth $\Omega(k)$ with $\omega(n^{1+1/k})$ edges, where all edges have weight $c$ except for one  \emph{special} edge $e^{*} = (u^{*}, v^{*})$ chosen uniformly at random, whose weight is 1. The metric $\sfd$ is the usual shortest path distance. Then $\sfd(u^{*}, v^{*}) = 1$, while the large girth pushes every other route in $G$ to $\Omega(ck)$ length. This is where the factors multiply. Girth lower bounds the number of hops a detour needs and the weights fix the price of a hop. It is unlikely that the edge is found directly because there are not enough distance oracle calls, and an adversarial $c$-BANN can always hide $e^*$ as well (returning another incident edge when, say, querying $u^*$). 

However, simply \emph{hiding} the edge $e^* $is not enough. The spanner algorithm can lower the $u^{*}\to v^{*}$ distance without ever seeing $e^{*}$. Specifically, there may potentially exist a path from $u^*$ to $v^*$ in the spanner of weight $o(ck)$ since the constructed spanner does not have to be a subgraph of $G$ itself and may have edges between, e.g., neighbors of $u^*,v^*$. 

To overcome this challenge, we characterize such alternative paths in Lemma~\ref{lem:re-routing}. In particular, we show that a spanner without $e^{*}$ that routes $u^{*}$ to $v^{*}$ at weight $o(ck)$ must contain an edge between a pair whose shortest path in $G$ passes through $e^{*}$ and has (unweighted) length $o(k)$. Then let $S$ be the set of the edges between all such pairs and the edges along their shortest path in $G$. Then the contrapositive of the above implies as long as $e^*$ is not in $S$, then the spanner cannot re-route $u^*$ to $v^*$ with distortion $o(ck)$. The spanner can only contain $o(n^{1+1/k}/k)$ edges and each shortest path is of length $o(k)$, so $|S| \leq o(n^{1+1/k}/k) \cdot k = o(n^{1+1/k})$. Since $e^*$ is uniformly distributed over $\Omega(n^{1+1/k})$ edges of $G$, one can expect $e \not \in S$ with high probability.



Finally, two sources of randomness, from $e^{*}$ and the spanner algorithm,  remain to be properly handled. The algorithm is adaptive, so its transcript depends on $e^{*}$. We therefore first assume the algorithm is deterministic and run it against the generic policy, which answers as though every edge weighed $c$, and plant $e^{*}$ outside $S$ only afterwards (with high probability). Then the generic run and the true run cannot be told apart, which enables us to apply the above analysis and lower bound the distortion of deterministic algorithms. Yao's principle allows us to extend the lower bound to randomized algorithms.

\subsection{Organization of the Paper}

Section \ref{sec: preliminaries} contains relevant definitions and preliminaries that are used throughout the paper. Section \ref{sec: fast-spanners} contains the algorithm and proof of Theorem \ref{thm: upper-bound}. Section \ref{sec: ck-lower-bound} contains the proof of Theorem \ref{thm: lower-bound}. Section \ref{sec: wasserstein-application} outlines the application of Theorem \ref{thm: upper-bound} to fast approximations of Wasserstein-$q$ distance, containing the proof of Corollary \ref{cor: wasserstein-application}.

\section{Preliminaries} \label{sec: preliminaries}

\subsection{Spanners}

\begin{definition}[$t$-Spanner]
    Let $\calM = ([n], \sfd)$ be a metric space. Graph $H = ([n], E_H)$ is a $t$-spanner, for $t \geq 1$, if 
    \begin{align*}
        \forall i, j \in [n]. \; \; \sfd_H(i, j) \leq t \cdot \sfd(i, j)
    \end{align*}
    where $\sfd_H$ denotes the shortest path metric in $H$.  
\end{definition}

\begin{definition}[$k$-Hop $t$-Spanner]
    Let $\calM = ([n], \sfd)$ be a metric space. Graph $H = ([n], E_H)$ is a $k$-hop $t$-spanner, for $k, t \geq 1$, if 
    \begin{align*}
        \forall i, j \in [n]. \; \; \exists (i, u_1), \ldots, (u_{k-1}, j) \in E_H. \; \; \sfd_H(i, u_1) + \sum_{\ell =1}^{k-2} \sfd_H(u_{\ell}, u_{\ell+1}) + \sfd_H(u_{k-1}, j) \leq t \cdot \sfd(i, j)
    \end{align*}
    That is, for any $i, j \in [n]$, there exists a path of length exactly $k$ between them in $H$ whose total distance is at most $t$ times that of the original distance $\sfd(i, j)$. Here, $\sfd_H$ denotes the shortest path metric in $H$. 
\end{definition}

\subsection{Approximate Nearest Neighbor Search}

In this work, we formulate Approximate Nearest Neighbor search as an online problem, while all other problems are thought to be offline. 

\begin{definition}[Approximate Nearest Neighbor Search]
    Let $\calM$ be a set of $n$ points in a metric space $(\calX, \sfd)$ and let $c \geq 1$. Then, the Approximate Nearest Neighbor Search problem asks to design a data structure $\calS$ which for any target set $T \subseteq \calM$, on any query $q \in \calX$, returns a point $p' \in T$ such that
    \begin{align*}
        \sfd(q, p') \leq c \cdot \min_{p \in T} \sfd(q, p)
    \end{align*}
    We let $\rho_{\calM}$ and $\kappa_{\calM}$ denote the preprocessing and query times of $\calS$ respectively. 
\end{definition}

\begin{definition}[Batch Approximate Nearest Neighbor Search]
    Let $\calM$ be a set of $n$ points in a metric space $(\calX, \sfd)$ and let $c \geq 1$. Then, the Batch Approximate Nearest Neighbor Search problem asks to design an algorithm $\calA$ which, on any target set $T \subseteq \calM$ and query set $Q \subseteq \calX$, returns the set of tuples $\{ (q, p', \sfd(q, p') ) \}_{q \in Q}$ such that
    \begin{align*}
        \forall q \in Q. \; \; \sfd(q, p') \leq c \cdot \min_{p \in T} \sfd(q, p)
    \end{align*}
    We let $\tau_{\calM}$ denote the runtime of $\calA$.  
\end{definition}

As illustrated in Figure \ref{fig: landscape}, one can use data structures for ANN to construct an algorithm for BANN. That is, given an ANN data structure with preprocessing time $\rho_{\calM}$ and query time $\kappa_{\calM}$, there trivially exists an algorithm for BANN with runtime $\bnntime = \rho_{\calM} + |Q| \cdot \kappa_{\calM}$. 

 We note $\bnntime$ is a function of the sizes of the query and target sets. Throughout this work, our algorithms will only ever invoke a BANN algorthm with query set $Q \subseteq \calM$. We make the natural assumption that $\bnntime$ is a monotonic function. Therefore, we write all our runtimes in terms of $\bnntime(n)$ where we have upper bounded $|Q|, |T| \leq n$. 

 \subsection{Wasserstein-$q$ distance}

We've already defined the Wasserstein-$q$ distance in Definition \ref{def: wq}. Here, we state and prove a standard lemma about $\sfW_q$ distance that will be useful in the proof of Corollary \ref{cor: wasserstein-application}. 
 
\begin{lemma} \label{lem: wq-sandwich}
    Let $\calM = ([n], \sfd)$ be a metric space, let $U \geq 1$ be an integer, and let $\mu, \nu \in \frac{1}{U} \Z^n$ be distributions on $[n]$. Then, for every $q \geq 1$,
    \begin{align*}
        \sfW_q(\mu, \nu) \; \leq \; \sfW_{\infty}(\mu, \nu) \; \leq \; U^{1/q} \cdot \sfW_q(\mu, \nu)
    \end{align*}
    In particular, taking $q = \log U$ gives $\sfW_{\log U}(\mu, \nu) \leq \sfW_{\infty}(\mu, \nu) \leq 2 \cdot \sfW_{\log U}(\mu, \nu)$.
\end{lemma}
 
\begin{proof}
    We write $\mathrm{supp}(\Gamma) = \{ (i,j) : \gamma_{ij} > 0 \}$. Then
    \begin{align*}
        \sfW_q(\mu, \nu)^q &= \min_{\gamma \in \Gamma(\mu, \nu)} \; \sum_{i, j \in [n]} \gamma_{ij} \cdot \sfd(i,j)^q \\
        \sfW_{\infty}(\mu, \nu) &= \min_{\gamma \in \Gamma(\mu, \nu)} \; \max_{(i,j) \in \mathrm{supp}(\gamma)} \sfd(i,j)
    \end{align*}
    For the first inequality, let $\gamma$ be optimal for $\sfW_{\infty}$. Since $\gamma$ is a distribution on $[n]^2$ and $\sfd(i,j) \leq \sfW_{\infty}(\mu, \nu)$ for every $(i,j) \in \mathrm{supp}(\gamma)$,
    \begin{align*}
        \sfW_q(\mu, \nu)^q \; \leq \; \sum_{i, j \in [n]} \gamma_{ij} \cdot \sfd(i,j)^q \; \leq \; \max_{(i,j) \in \mathrm{supp}(\gamma)} \sfd(i,j)^q \; = \; \sfW_{\infty}(\mu, \nu)^q
    \end{align*}
 
    For the second inequality, we use the following integrality fact: since $U\mu, U\nu \in \Z_{\geq 0}^n$, the transportation problem with supplies $U\mu$ and demands $U\nu$ has an integral optimal solution by the integrality of minimum-cost flow \cite[Theorem 9.10]{AMO93}. Scaling by $U$, there is an optimal coupling $\gamma^{\star}$ for $\sfW_q$ with $\gamma^{\star} \in \frac{1}{U} \Z^{n \times n}$. Let $(i^{\star}, j^{\star}) \in \mathrm{supp}(\gamma^{\star})$ maximize $\sfd$ over $\mathrm{supp}(\gamma^{\star})$. Then, 
    \begin{align*}
        \sfW_q(\mu, \nu)^q = \sum_{i, j \in [n]} \gamma^{\star}_{ij} \cdot \sfd(i,j)^q \; \geq \; \gamma^{\star}_{i^{\star} j^{\star}} \cdot \sfd(i^{\star}, j^{\star})^q \; \geq \; \frac{1}{U} \cdot \sfd(i^{\star}, j^{\star})^q \; \geq \; \frac{1}{U} \cdot \sfW_{\infty}(\mu, \nu)^q
    \end{align*}
    where the final inequality holds because $\gamma$ is a feasible coupling for $\sfW_{\infty}$, whose objective value on $\gamma$ is exactly $\sfd^{\star}$. Taking $q$-th roots gives $\sfW_{\infty}(\mu, \nu) \leq U^{1/q} \cdot \sfW_q(\mu, \nu)$. Finally, for $q = \log_2 U$ we have $U^{1/q} = 2$.
\end{proof}

\newcommand{\otime}{\mathrm{T}_{\calO}}

\section{Proof of Theorem \ref{thm: upper-bound}} \label{sec: fast-spanners}

Our algorithm follows the description provided in the technical overview (Section \ref{sec: technical-overview}). We begin by proving that for every scale $r$, the shortest path metric ($\sfd_G$) in the spanner $G = ([n], E_r)$ returned by Algorithm \ref{alg: fast-spanners} has stretch at most $O(ck)$. The proof of Theorem \ref{thm: upper-bound} then proceeds by performing a careful accounting of the size of the resulting spanner and the total construction time. Throughout the proof, we will assume the oracle $\calO$ reports correctly. We will union bound over the failure probability at the very end, in the proof of Theorem \ref{thm: upper-bound}. 

For a subset $P \subseteq [n]$ and point $v \in [n]$, we let $\sfd(v, P) := \min_{z \in P} \sfd(v, z)$. For a given $\ell \in [\log \Delta]$, and an active point $v \in [n]$, we let $C_i(v)$ denote the active cluster containing $v$ in the $i$th iteration. Every active point $v$ additionally maintains a \emph{candidate list} $L_v$ of pairs $(C, z)$ with $C \in \calC_{\ell, i}$ and $z \in C$. We write $C \in L_v$ for the clusters occurring in it, and we call $L_v$ \emph{full} when $|L_v| = \tau$, for some $\tau$ to be set later. Finally, for a hash function $h$ and a set of clusters $F$, we write $h(F) := \{ h(C) ~|~ C \in F\}$.


\begin{algorithm}
    \caption{BANN-Spanner}
    \label{alg: fast-spanners}
    \KwIn{$n, c, k, \Delta, \calO$}

    Initialize $\tau \gets \kappa n^{1/(2k)} \log (10^3 nk \log \Delta), B \gets \kappa \tau, H \gets \kappa \tau$ for $\kappa$ a sufficiently large constant \\
    ~\ \\
    \ForEach{$\ell \in [\log \Delta]$}{
        Initialize $r \gets 2^{\ell}$, $E_r \gets \emptyset$, $A_{\ell, 0} \gets [n]$, $\calC_{\ell, 0} \gets \{ \{x\} ~|~ x \in [n]\}$ with $\mathrm{Center}(\{x \}) = x$ \\
        ~\ \\
        \ForEach{$i \in \{0, \ldots, 2k-2\} $}{
            Let $L_v \gets \emptyset$ for every $v \in A_{\ell, i}$ \\
            ~\ \\


            \ForEach{$t \in [H]$}{ \label{line: hashing-start}
                Hash the active clusters $\calC_{\ell, i}$ via uniformly random hash function $h \sim \calH$ into $B$ buckets \\
                Let $F_v \gets \{ C_{\ell, i}(v) \} \cup \{ C ~|~ C \in L_v \}$ for every $v \in A_{\ell, i}$ \\
                ~\ \\
                \ForEach{$j \in [B]$}{
                    $X_j \gets \{ u \in A_{\ell, i} ~|~ h(C_{\ell, i}(u)) = j \}$ \\
                    $Q_j \gets \{ v \in A_{\ell, i} ~|~ |L_v| < \tau, ~ j \notin h(F_v) \}$ \\
                    ~\ \\
                    $\{ (v, u, \sfd(v, u)) \}_{v \in Q_j} \gets \calO(Q_j, X_j, c)$ \\
                    ~\ \\
                    \ForEach{$v \in Q_j$}{
                        \If{$\sfd(v, u) \leq cr$ and $|L_v| < \tau$}{
                            Append $(C_{\ell, i}(u), u)$ to $L_v$ \label{line: hashing-end}
                        }
                    }
                }
            }
            ~\ \\
            Sample clusters from $\calC_{\ell, i}$ into $\calS_{\ell, i}$ independently with probability $n^{-1/(2k)}$ \label{line: sampling} \\
            Let $\calC_{\ell, i+1} \gets \calS_{\ell, i}$ \\
            Let $\calU_{\ell, i} \gets \calC_{\ell, i} \setminus \calS_{\ell, i}$ denote the set of unsampled clusters \\
            ~\ \\
            Let $S_{\ell, i} \gets \{ v \in A_{\ell, i} ~|~ C_{\ell, i}(v) \in \calS_{\ell, i} \}$ denote the set of sampled vertices\\
            Let $U_{\ell, i} \gets A_{\ell, i} \setminus S_{\ell, i}$ denote the set of active but unsampled vertices \\
            Let $Z_{\ell, i} \gets \emptyset$ denote the set of retired vertices \\
            ~\ \\
            \ForEach{$v \in U_{\ell, i}$}{
                \If{$(C, z) \in L_v$ for some $C \in \calS_{\ell, i}$}{
                    Add edge $(v, z)$ with weight $\sfd(v, z)$ to $E_r$ \label{line: attach-edge} \\
                    Add $v$ to the cluster $C \in \calC_{\ell, i+1}$ \label{line: join-to-cluster}
                }
                \Else{
                    $Z_{\ell, i} \gets Z_{\ell, i} \cup \{v\}$ \\
                    Add edge $(v, z)$ with weight $\sfd(v, z)$ to $E_r$ for every $(C, z) \in L_v$ \label{line: retire-edges}
                }
            }
            ~\ \\
            Let $A_{\ell, i+1} \gets A_{\ell, i} \setminus Z_{\ell, i}$ \\
        }
        ~\ \\
        \ForEach{$C \in \calC_{\ell, 2k-1}$}{
            $\{ (v, u , \sfd(v, u)) \}_{v \in A_{\ell, 2k-1} \setminus C} \gets \calO(A_{\ell, 2k-1}\setminus C, C, c)$ \\
            ~\ \\
            \ForEach{$v \in A_{\ell, 2k-1} \setminus C$}{
                \If{$\sfd(v, u) \leq cr$}{
                    Add edge $(v, u)$ with weight $\sfd(v, u)$ to $E_r$
                }
            }
        }
    }
    ~\ \\
    Return $\bigcup_{r \in \{ 2^\ell \}_{\ell=1}^{\log \Delta}} E_{r}$
\end{algorithm}

We assume, without loss of generality, that the minimum distance $\min_{x \neq y \in [n]} \sfd(x, y) = 1$ by scaling. We define event $\calE_1$ to be the event that for all $\ell \in [\log \Delta]$, $i \in \{0, \ldots, 2k-2\}$, and $v \in A_{\ell, i}$, the list $L_v$ produced in Lines \ref{line: hashing-start} - \ref{line: hashing-end} is either full or contains every cluster $C \in \calC_{\ell, i}$ with $C \neq C_{\ell, i}(v)$ and $\sfd(v, C) \leq r$. The subsequent lemma proves that event $\calE_1$ occurs with high probability.

\begin{lemma} \label{lem: list-completeness}
    Event $\calE_1$ holds with probability at least $0.999$.
\end{lemma}
\begin{proof}
    Fix an $\ell \in [\log \Delta]$, an iteration $i \in \{0, \ldots, 2k-2\}$, and a vertex $v \in A_{\ell, i}$. Call $v$ \emph{unfinished} at the start of a repetition $t \in [H]$ if $L_v$ is not full and some cluster $C \neq C_{\ell, i}(v)$ with $\sfd(v, C) \leq r$ is missing from $L_v$. We first argue that, conditioned on any history of the preceding repetitions in which $v$ is unfinished, $|L_v|$ strictly increases during repetition $t$ with probability at least $3/4$.

    Fix a missing cluster $C_{\star} \neq C_{\ell, i}(v)$ with $\sfd(v, C_{\star}) \leq r$, and observe that $C_{\star} \notin F_v$ while $|F_v| \leq \tau + 1$. Because $h$ is drawn afresh in repetition $t$,
    \begin{align*}
        \Pr \left[ h(C_{\star}) \notin h(F_v) ~\middle|~ \text{iterations } 1, \ldots, t-1 \right] \geq 1 - \frac{|F_v|}{B} \geq 1 - \frac{\tau + 1}{\kappa \tau} \geq \frac{3}{4}
    \end{align*}
    Write $j = h(C_{\star})$ and suppose that this event occurs. If $L_v$ became full earlier in repetition $t$ then $|L_v|$ has already increased and there is nothing to prove, so assume otherwise; then $v \in Q_j$. The bucket $X_j$ contains $C_{\star}$ and is in particular nonempty, so oracle $\calO$ is called and returns a $u \in X_j$ with $\sfd(v, u) \leq c \cdot \sfd(v, X_j) \leq c \cdot \sfd(v, C_{\star}) \leq cr$, and the pair $(C_{\ell, i}(u), u)$ is appended to $L_v$. This pair is new: $h(C_{\ell, i}(u)) = j \notin h(F_v)$, so $C_{\ell, i}(u)$ was not in $L_v$ at the start of repetition $t$, and any cluster appended to $L_v$ earlier in repetition $t$ was drawn from a bucket other than $j$ and so does not hash to $j$. Hence $|L_v|$ increases, as claimed.

    Let $Y_t$ indicate that repetition $t$ increased $|L_v|$, or that $v$ was already finished at its start. By the claim, $Y_t = 1$ with probability at least $3/4$ conditioned on any history of the preceding repetitions. Therefore, 
    \begin{align*}
        \forall m \in [H]. \; \; \Pr \left[ \sum_{t=1}^{H} Y_t < m \right] \leq \Pr \left[ \text{Bin}(H, 3/4) < m \right]
    \end{align*}
    by a simple inductive argument. If $v$ is unfinished after all $H$ repetitions then $\sum_{t} Y_t < \tau$, as $|L_v| \leq \tau$ throughout. Therefore, since $H = \kappa \tau \geq 4 \tau$,
    \begin{align*}
        \Pr \left[ v \text{ unfinished after } H \text{ repetitions} \right] \leq \Pr \left[ \text{Bin}(H, 3/4) < \tau \right] \leq \exp \left( - \Omega(H) \right) \leq \frac{1}{10^3 nk \log \Delta}
    \end{align*}
    where the final step uses $\tau \geq \kappa \log(10^3 nk \log \Delta)$. Taking a union bound over the at most $2nk \log \Delta$ many tuples $(\ell, i, v)$ proves the claim.
\end{proof}

We now define event $\calE_2$ to be the event that for all $\ell \in [\log \Delta]$, $i \in \{0, \ldots, 2k-2\}$, and $v \in A_{\ell, i}$ whose list $L_v$ is full, some cluster of $L_v$ is sampled into $\calS_{\ell, i}$.

\begin{lemma} \label{lem: full-lists-sampled}
    Event $\calE_2$ holds with probability at least $0.999$.
\end{lemma}
\begin{proof}
    Fix an $\ell \in [\log \Delta]$, an iteration $i \in \{0, \ldots, 2k-2\}$, and a vertex $v \in A_{\ell, i}$ whose list is full. The lists are built in Lines \ref{line: hashing-start} - \ref{line: hashing-end}, which precede the sampling of Line \ref{line: sampling}. Conditioned on $L_v$, the $\tau$ clusters it contains are distinct members of $\calC_{\ell, i}$, each of which is placed into $\calS_{\ell, i}$ independently with probability $n^{-1/(2k)}$. Thus,
    \begin{align*}
        \Pr[\exists ~ \ell \in [\log \Delta], i \in \{0, \ldots, 2k-2\}, v \in A_{\ell, i} &: L_v \text{ is full and } L_v \cap \calS_{\ell, i} = \emptyset] \\
        &\leq \sum_{\ell = 1}^{\log \Delta} \sum_{i = 0}^{2k-2} \sum_{v \in A_{\ell, i}} \Pr \left[ L_v \cap \calS_{\ell, i} = \emptyset ~\middle|~ L_v \text{ is full} \right] \\
        &\leq \sum_{\ell} \sum_{i} \sum_{v \in A_{\ell, i}} \left( 1 - n^{-1/(2k)} \right)^{\tau} \\
        &\leq 2 n k \log \Delta \cdot \left( \frac{1}{10^3 nk \log \Delta} \right)^{\kappa} \leq 0.001
    \end{align*}
    thereby proving the claim.
\end{proof}

The two events together imply Lemma \ref{lem: hashing-correctness}, which shows that retired vertices will be connected to \emph{all} clusters at distance at most $r$. 

\begin{lemma} \label{lem: hashing-correctness}
    Conditioned on $\calE_1 \cap \calE_2$, the following holds for every $\ell \in [\log \Delta]$, $i \in \{0, \ldots, 2k-2\}$, $v \in Z_{\ell, i}$, and cluster $C \in \calC_{\ell, i}$ with $C \neq C_{\ell, i}(v)$ and $\sfd(v, C) \leq r$. The edge set $E_r$ contains an edge $(v, z)$ where $z \in C$ and $\sfd(v, z) \leq cr$.
\end{lemma}
\begin{proof}
    Fix any such tuple of $(\ell, i, v, C)$. Because $v \in Z_{\ell, i}$, no cluster of $L_v$ was sampled into $\calS_{\ell, i}$, and so $L_v$ is not full by event $\calE_2$. Event $\calE_1$ then guarantees that $C \in L_v$, that is, $(C, z) \in L_v$ for some $z \in C$ with $\sfd(v, z) \leq cr$. Since $v$ is retired, Line \ref{line: retire-edges} adds this edge to $E_r$, thereby proving the claim.
\end{proof}

\begin{lemma} \label{lem: close-to-center}
    Let $\ell \in [\log \Delta], i \in \{0, \ldots, 2k-1\}$ and consider any cluster $C \in \calC_{\ell, i}$. Then, $\sfd_G(v, \mathrm{Center}(C)) \leq icr$ for every $v \in C$.
\end{lemma}
\begin{proof}
    We proceed by induction. Consider a fixed $\ell \in [\log \Delta]$. The base case is when $i = 0$. Here, $\calC_{\ell, 0} = \{ \{x \} ~|~ x \in X\}$, so the claim is trivially satisfied.

    For the inductive step, assume that the claim holds for $i > 0$, then we show that it also holds for $i+1$. Consider a cluster $C' \in \calC_{\ell, i+1}$ produced by a sampled cluster $C \in \calS_{\ell, i}$. A point $v \in C'$ was either already in $C$ or joined to $C'$ in Lines \ref{line: attach-edge} - \ref{line: join-to-cluster} in iteration $i$. If $v \in C$, then $\sfd_G(v, \mathrm{Center}(C')) = \sfd_G(v, \mathrm{Center}(C)) \leq icr \leq (i+1)cr$ by the inductive hypothesis. Otherwise, if $v$ joins $C'$, it must have been the case that $(C, u) \in L_v$ for some point $u \in C$, and every pair appended to a candidate list satisfies $\sfd(v, u) \leq cr$. Therefore,
    \begin{align*}
        \sfd_G(v, \mathrm{Center}(C')) &\leq \sfd_G(v, u) + \sfd_G(u, \mathrm{Center}(C')) \\
        &= \sfd(v, u) + \sfd_G(u, \mathrm{Center}(C)) \\
        &\leq c r + i c r \\
        &= (i+1)c r
    \end{align*}
    where the second inequality invokes the inductive hypothesis to conclude that $\sfd_G(u, \mathrm{Center}(C)) \leq icr$. 
\end{proof}

\begin{lemma} \label{lem: distortion}
    Conditioned on $\calE_1 \cap \calE_2$, for every $\ell \in [\log \Delta]$ and $r = 2^{\ell}$, any pair of points $x, y \in [n]$ with $\sfd(x, y) \leq r$ is joined in $E_r$ by a path of at most $4k-1$ edges and total weight at most $(4k-1)cr$. In particular, $\sfd_G(x, y) \leq (4k-1)c r$.
\end{lemma}
\begin{proof}
    Fix an $\ell \in [\log \Delta]$ with $r = 2^{\ell}$, and a pair $x, y \in [n]$ with $\sfd(x, y) \leq r$. Note first that the induction of Lemma \ref{lem: close-to-center} prepends a single edge per level, so for $C \in \calC_{\ell, i}$ and $v \in C$ the path it guarantees from $v$ to $\mathrm{Center}(C)$ uses at most $i$ edges. Consequently any two points of a common cluster in $\calC_{\ell, i}$ are joined through its center by at most $2i \leq 4k-2$ edges of total weight at most $2icr \leq (4k-2)cr$.

    If there exists an $i \in \{0, \ldots, 2k-2\}$ for which $x, y$ lie in the same active cluster in $\calC_{\ell, i}$, then by Lemma \ref{lem: close-to-center}, we have $\sfd_G(x, y) \leq (4k-2)cr$.

    Otherwise, consider the first iteration $i$ when at least one of $x, y$ is \emph{retired} into $Z_{\ell, i}$. Without loss of generality, assume that $x \in Z_{\ell, i}$. We may assume that $C_{\ell, i}(y) \neq C_{\ell, i}(x)$, since otherwise Lemma \ref{lem: close-to-center} applies as above. As $\sfd(x, C_{\ell, i}(y)) \leq \sfd(x, y) \leq r$, Lemma \ref{lem: hashing-correctness} guarantees that $E_r$ contains an edge $(x, z)$ where $z \in C_{\ell, i}(y)$ and $\sfd(x, z) \leq cr$. Therefore,
    \begin{align*}
        \sfd_G(x, y) &\leq \sfd_G(x, z) + \sfd_G(z, y) \\
        &\leq \sfd(x, z) + 2(2k-2)c r \tag{Lemma \ref{lem: close-to-center}} \\
        &\leq (4k-3)cr
    \end{align*}
    realized by the path $x \to z \to \mathrm{Center}(C_{\ell, i}(y)) \to y$, of at most $1 + 2(2k-2) \leq 4k-1$ edges.

    The final case occurs if both $x, y$ are never retired in the first $2k-1$ iterations. If $C_{\ell, 2k -1}(y) = C_{\ell, 2k-1}(x)$, then Lemma \ref{lem: close-to-center} applies. Otherwise, the final step of the algorithm exhaustively calls the BANN oracle $\calO$ on the final clusters/remaining active vertices. Because $x \notin C_{\ell, 2k-1}(y)$, oracle $\calO$ will return a point $z \in C_{\ell, 2k-1}(y)$ with $\sfd(x, z) \leq cr$ and add the corresponding edge. Therefore, $\sfd_G(x, y) \leq \sfd_G(x, z) + \sfd_G(z, y) \leq cr + 2(2k-1)cr = (4k-1)cr$ as before, along a path of at most $1 + 2(2k-1) = 4k-1$ edges.
\end{proof}

\begin{lemma}\label{lem: final-clusters}
    With probability at least $0.999$, for all $\ell \in [\log \Delta]$, the final number of active clusters $|\calC_{\ell, 2k-1}| \leq 10 n^{1/(2k)} + 10 \log (1000 \log \Delta)$.
\end{lemma}
\begin{proof}
    Fix a scale $\ell \in [\log \Delta]$. Then, observe that $|\calC_{\ell, 2k-1}| \sim \text{Bin}(n, n^{(1-2k)/(2k)})$. By a Chernoff bound then,
    \begin{align*}
        \Pr \left[ |\calC_{\ell, 2k-1}| \geq 10 n^{1/(2k)} + 10 \log (10^3 \log \Delta) \right] &\leq \exp \left( - \log (10^3 \log \Delta) \right) = \frac{1}{10^3 \log \Delta}
    \end{align*}
    Taking a union bound over all scales $\ell$ gives the final claim.
\end{proof}

We now have the necessary components to prove Theorem \ref{thm: upper-bound}.

\begin{proof}[Proof of Theorem \ref{thm: upper-bound}]
    We first count the number of calls to oracle $\calO$. For given scale $\ell \in [\log \Delta]$ and iteration $i \in \{0, \ldots, 2k-2\}$, Algorithm \ref{alg: fast-spanners} makes at most $HB$ calls to $\calO$, all of them in Lines \ref{line: hashing-start} - \ref{line: hashing-end}. Since $HB = \kappa^2 \tau^2 = \tilde{O}(n^{1/k})$, this results in $\tilde{O}(k n^{1/k} \log \Delta)$ calls to oracle $\calO$. Conditioned on the event of Lemma \ref{lem: final-clusters}, there are at most $O(n^{1/(2k)} + \log \log \Delta)$ final clusters. Then, Algorithm \ref{alg: fast-spanners} makes at most $O(n^{1/(2k)} + \log \log \Delta)$ final calls to oracle $\calO$ per scale. All remaining computation in Algorithm \ref{alg: fast-spanners} can be done in $O(n)$ time per call to $\calO$ (to construct the query and target sets), together with $O(nk \log \Delta)$ additional time. Because $\bnntime(n) = \Omega(n)$, the total runtime of Algorithm \ref{alg: fast-spanners} is therefore $\tilde{O}(\bnntime(n) \cdot k n^{1/k} \log \Delta)$.

    \textbf{Size. } The oracle calls of Lines \ref{line: hashing-start} - \ref{line: hashing-end} insert no edges, they only populate the lists $L_v$. For a fixed scale $\ell$ and iteration $i$, each $v \in U_{\ell, i}$ contributes a single edge in Line \ref{line: attach-edge} if it attaches and at most $|L_v| \leq \tau$ edges in Line \ref{line: retire-edges} if it retires, so the iteration inserts at most $O(n \tau k)$ edges. Conditioned on the event of Lemma \ref{lem: final-clusters}, the final round inserts at most $O(n^{1 + 1/(2k)} + n \log \log \Delta)$ edges per scale. Therefore, the size of the spanner is bounded by $\tilde{O}(k n^{1 + 1/(2k)} \log \Delta)$. 
    

    \textbf{Hops and Distortion. } Consider a pair $x, y \in [n]$ and let $\ell^{\star}$ denote the scale at which $\sfd(x, y) \in [2^{\ell^{\star}}, 2^{\ell^{\star} + 1}]$. Then, by Lemma \ref{lem: distortion}, $\sfd_G(x, y) \leq (4k-1)c \cdot 2^{\ell^{\star} + 1}$. Furthermore, because Algorithm \ref{alg: fast-spanners} constructs spanner $G$ by weighting any inserted edges by only true metric distances, it is clear that $\sfd_G(x, y) \geq \sfd(x, y)$. Therefore,
    \begin{align*}
        \forall x, y \in [n]. \; \; \sfd(x, y) \leq \sfd_G(x, y) \leq O(c k) \cdot \sfd(x, y)
    \end{align*}
    $G$ is a $(4k-1)$-hop spanner by Lemma \ref{lem: distortion}.

    \textbf{Success Probability.} Algorithm \ref{alg: fast-spanners} succeeds when events $\calE_1, \calE_2$ succeed, all calls to the BANN succeed, and the event of Lemma \ref{lem: final-clusters} holds. Note that in the definition of the BANN oracle (Definition \ref{def: bann}), we assume that the oracle has failure probability $\leq 1/n$ (which can always be achieved via independent repetitions of an algorithm with constant failure probability). By a union bound then, Algorithm \ref{alg: fast-spanners} fails with probability at most $0.01$.
\end{proof}

\section{Proof of Theorem \ref{thm: lower-bound}} \label{sec: ck-lower-bound}
In this section, we prove an $\Omega(ck)$ lower bound for the distortion factor of spanners constructed with $c$-BANN oracles and query budget $\Phi \coloneqq o(n^{1+1/k} / k)$. We first fix the model for spanner construction algorithms. 

\paragraph{Oracle query model} An adversary fixes a metric $\calM = ([n], \sfd)$. An algorithm $\calA$ interacts with it adaptively through two oracles.
    \begin{itemize}
        \item \textbf{$c$-BANN oracle.} A call $\calO(Q, X, c)$ with $Q, X \subseteq [n]$ produces an unordered answer list with one entry per query point, the entry for $q \in Q$ being $(q, \varphi(q), \sfd(q, \varphi(q)))$, where $\varphi(q) \in X$ is \emph{any} point satisfying
        \[
            \sfd(q, \varphi(q)) \;\leq\; c \cdot \min_{x \in X \setminus \{q\}} \sfd(q, x).
        \]
        The call costs $|Q|$.
        \item \textbf{Distance oracle.} On a call $\calO_{\sfd}(x, y)$ the oracle returns $\sfd(x,y)$ exactly. The call costs $1$.
    \end{itemize}
     
    \paragraph{The spanner algorithm $\calA$.} The two oracles above are $\calA$'s only access to the metric $\calM$ and it is also given $n$, $c$ and $k$. $\calA$ has a total budget $\Phi$ for cost of oracle calls. It outputs a weighted spanner for $[n]$: $H = \left([n], E_H \subseteq \binom{[n]}{2}, w_H: E_H \to \mathbb{R}\right)$. A spanner may not contract distances, so $w_H(e = (x,y)) \geq \sfd(x,y)$ on every edge, and lowering a weight to $\sfd(x,y)$ only helps $\calA$. Therefore, WLOG, $w_H(x,y) = \sfd(x,y)$. Thus $\sfd_H \geq \sfd$ always, and the distortion of $H$ is
    \[
        t \;=\; \max_{x \neq y} \frac{\sfd_H^w(x,y)}{\sfd(x,y)},
    \]
    where $\sfd_H^w$ is the shortest path metric of $H$. We assume, as part of the model, that $\calA$ may include a pair as an edge of $H$ only if it has learned that pair's distance from one of the two oracles; writing $T$ for the set of pairs named by the answers $\calA$ has received, this reads $E_H \subseteq T$. This is the natural formalisation of constructing a spanner from oracle queries: without it $\calA$ may output edges whose weights it cannot certify. Each edge of $H$ therefore costs $\calA$ at least one unit of budget, and $|E_H| \leq \Phi$. 

\begin{theorem}[Lower Bound] \label{thm: ck-lower-bound}
    Let $c > 1$, let $k \geq 2$, and let $\Phi = \Phi(n)$ be a budget with
    \[
        \Phi(n) \;=\; o\!\left( \frac{n^{1+1/k}}{k} \right).
    \]
    Then for every deterministic algorithm $\calA$ of budget at most $\Phi$ there is an instance on which $\calA$ incurs distortion $t = \Omega(ck)$; and for every randomized algorithm of budget at most $\Phi$ there is an instance on which its expected distortion, over its own randomness, is $\Omega(ck)$. 
\end{theorem}

\subsection{The instance and the adversary's policy}
\label{sec:instance-and-policy}
To construct the instance for our lower bound, we need a graph with high girth in the following classic result. 

\begin{lemma}[{\protect\cite{erdos1963regulare}}]
    \label{lem:girth-lower}
    For any integer $\lambda\geq 3$ and $n$ sufficiently large, there exists a graph over $n$ nodes with girth at least $\lambda$ and $\Omega\left(n^{1+1/(\lambda-2)}\right)$ edges.
\end{lemma}

\paragraph{The instance for the lower bound} Fix $c > 1$ and an integer $k \geq 2$, and put $k' \coloneqq \lceil k/2 \rceil$. Let weighted graph $G = ([n], E, w)$ be the following: Starting from a graph with girth $\geq k+2$ and $\Omega(n^{1+1/k})$ edges, which Lemma~\ref{lem:girth-lower} with $\lambda = k+2$ supplies for $n$ large enough and whose girth is $\geq 2k' + 1$ ($k + 2 \geq 2k' + 1$ since $k' = \lceil k/2\rceil$), sample one special edge $e^{*} = (u^{*}, v^{*})$ uniformly from $E$ and assign weights
\[
    w(e) \;=\; \begin{cases} 1 & e = e^* \\ c & \text{otherwise,} \end{cases}
\]
let $\sfd_G^w$ be the shortest path metric of $G$, and let the instance metric be the truncation:
\[
    \sfd(x, y) \;:=\; \min\{ \sfd_G^w(x,y), \; \Lambda \}, \text{where } \Lambda := k'c.
\]
Truncating a metric at a constant preserves the triangle inequality, so $\sfd$ is indeed a metric. Since only one edge is light, a path or a cycle with $s$ edges weighs at least $(s-1)c + 1$; we use this repeatedly below. In particular $\sfd(u^{*}, v^{*}) = 1$, as any other route between $u^{*}$ and $v^{*}$ has at least two edges, none of them $e^{*}$, and so weighs at least $2c$. Note that the instance has aspect ratio exactly $\Lambda = \Theta(ck)$, so no algorithm can be worse than the bound we prove; this is inherent to girth-based constructions, and the point is that a budget-bounded algorithm is forced all the way up to it.

\paragraph{The adversary's policy} Let $\sfd_G$ be the (unweighted) graph distance function of $G$. For a query point in the query set $q \in Q$ with target set $X$:
\begin{itemize} 
    \item If $q \in X$, return $q$ and report distance $0$
    \item Otherwise, let $j:=\min_{x\in X}\sfd_G(q,x)\geq 1$. If $j = 1$, scan $q$'s neighbors in $X$ in an arbitrary fixed order and return the first one joined to $q$ by an edge of weight $c$, if one exists. Otherwise, return the target neighbor joined by $q$ by the special edge of distance $1$. If $j \geq 2$, return a point of $X$ that is closest to $q$ in $\sfd_G$, breaking ties according to an arbitrary fixed order.
\end{itemize}

These answers are legal for every $c > 1$: the returned point is at distance at most $\min\{jc, \Lambda\}$ from $q$ while every point of $X$ is at distance at least $\min\{(j-1)c + 1, \Lambda\}$ from $q$, and $\min\{jc, \Lambda\} \leq c \cdot \min\{(j-1)c + 1, \Lambda\}$ since $c > 1$ and $jc \leq (j-1)c^2 + c$.

For all distance oracle queries, the adversary has no choice but to be honest.

\subsection{The proof of Theorem~\ref{thm: ck-lower-bound}}
One can expect edges in the output of the $c$-BANN oracle and the distance oracle calls made by the spanner algorithm $\calA$ to miss the special edge $e^* = (u^*, v^*)$ of weight $1$ with high probability if $\calA$ respects the budget. The following lemma states that some edge with desired properties must exist to make up for missing the special edge assuming the spanner $H$ output by $\calA$ still has small distortion.  

\begin{lemma}
\label{lem:re-routing}
    Suppose the output spanner $H$ satisfies $\sfd_H(u^*, v^*) < \Lambda$. Then there must exist some edge $(x,y) \in E_H$ that has both of the following properties:
    \begin{enumerate}[label=(\alph*)]
        \item $x$ and $y$ have a unique shortest path in $G$ and $\sfd_G(x,y) \leq k'$; and
        \item this shortest path runs through the edge $e^*$.
    \end{enumerate}
\end{lemma}

\begin{proof}
    If $e^{*} \in E_H$ the conclusion holds with $\{x,y\} = \{u^{*}, v^{*}\}$, so assume $e^{*} \notin E_H$. The proof alternates between the original graph $G$ and the spanner $H$. 
    
    Let $u^{*} = z_0, z_1, \dots, z_r = v^{*}$ be the cheapest weighted path between $u^*, v^*$ in $H$, so that each consecutive pair $\{z_i, z_{i+1}\}$ is an edge of $H$ and
    \[
        \sum_{i=0}^{r-1} \sfd(z_i, z_{i+1}) \;=\; \sfd_H(u^{*}, v^{*}) \;<\; \Lambda .
    \]
    Each term is at most the whole sum, so $\sfd(z_i, z_{i+1}) < \Lambda$ for every $i$. A path with $s$ edges weighs at least $(s-1)c + 1$, so a path of weight less than $\Lambda = k'c$ has $s < k' + 1 - 1/c$, hence $s \leq k'$; thus $\sfd_G(z_i, z_{i+1}) \leq k'$. Therefore, every pair $z_i,z_{i+1}$ has a unique shortest path in $G$ since, otherwise, there would be a cycle of weight less than $2\Lambda = 2k'c$ which indicates a cycle of $\leq 2k'$ edges in $G$ contradicting the girth of $G$. We will pick $x,y$ from pairs among $z_i,z_{i+1}$ for $i \in [0,r-1]$ which, according to the above, guarantees property (a) in the lemma statement. We name the shortest paths between $z_i,z_{i+1}$ in $G$ as follows:
    \[
        P_i \;:=\; P_{z_i z_{i+1}}.
    \]
    Concatenate $P_0, P_1, \dots, P_{r-1}$ end to end, i.e., the last point of $P_i$ and the first point of $P_{i+1}$ are the same point $z_{i+1}$. The resulting path in $G$ from $u^{*}$ to $v^{*}$ is denoted by $W$ . Note that $W$ could contain cycles and it has weight
    \[
        w(W) \;=\; \sum_{i=0}^{r-1} w(P_i) \;=\; \sum_{i=0}^{r-1} \sfd(z_i, z_{i+1}) \;=\; \sfd_H(u^{*}, v^{*}) \;<\; \Lambda .
    \]
    The $H$-path has now been left behind. The rest of the proof takes place entirely inside $G$.

    We now prove that $W$ has to traverse $e^{*}$ in $G$. Suppose it does not. Then $W$ together with $e^{*}$ is a closed walk traversing $e^{*}$ exactly once, and therefore contains a simple cycle through $e^{*}$, of weight at most $w(W) + 1 < \Lambda + 1$, which indicates the cycle has at most $\lceil(\Lambda + 1) /c \rceil = \lceil k' + 1/c \rceil = k'+1$ edges, using $c > 1$. This is a contradiction since $G$ has girth at least $2k'+1$. So $W$ does traverse $e^{*}$.

   We now locate $x$ and $y$. Walk along $W$ from $u^{*}$ until $e^{*}$ is traversed for the first time, and let $i$ be the index of the segment being traversed at that moment: $W$ was assembled out of $P_0, \dots, P_{r-1}$, and this says $e^{*}$ lies on $P_i$. That index names the pair,
    \[
        \{x,y\} \;:=\; \{z_i, z_{i+1}\} .
    \]
    This gives property (b).
\end{proof}

Re-routing the special edge in the spanner $H$ requires that $H$ contains some particular pair $(x,y)$ as stated in the above lemma. Now we define the following notation of $S$ such that if $e^* \in S(H)$, then some edge in $H$ can play the role of $x,y$ in lemma~\ref{lem:re-routing} and save $H$ from missing $e^*$.

\begin{definition}
Write $P_{xy}$ for the shortest path from $x$ to $y$ in $G$. It is well defined and does not depend on $e^{*}$: a path with $s$ edges weighs at least $(s-1)c + 1 > s'c$ whenever $s > s'$, so the weighted shortest paths are exactly the hop-shortest ones no matter which edge is special, and when the two points are at most $k'$ edges apart the girth of $G$ makes the hop-shortest path unique. Call a pair of points $a,b$ close if its endpoints are at most $k'$ edges apart in $G$ ($\sfd_G(a,b) \leq k'$). For any set $F$ of pairs of points, let $S(F)$ be the set of edges that the close pairs of $F$ have their shortest paths run through:
\[
    S(F) \;:=\; \big\{\, e \in E(G) \;:\; e \text{ lies on } P_{xy} \text{ for some close pair } \{x,y\} \in F \,\big\} \;\subseteq\; E(G) .
\]
\end{definition}
We will abuse the notation of $S(\cdot)$ and let it take in a graph or a set of edges. For example, $S(H)$ and $S(E_H)$ will both mean applying $S(\cdot)$ to the pairs indicated by the edge set of $H$. Note that $S$ is monotone, $F \subseteq F'$ implying $S(F) \subseteq S(F')$.

\begin{observation}
\label{obs:S-coverage}
For any graph $G = (V, E)$, and subset of pairs $F \subset V \times V$,
\[
     |S(F)| \leq |F| k' .
\]
\end{observation}

\begin{proof}
 A pair of $F$ that is not close contributes nothing to $S(F)$. A close pair contributes the edges of a single shortest path, because two shortest paths between points at most $k'$ edges apart would close a cycle of at most $2k'$ edges, contradicting the girth of $G$. Each of these paths contributes at most $k'$ edges by definition, so $|S(F)| \leq |F| k'$.
\end{proof}

Now we start taking a closer look at the spanner algorithm's behavior. We will show the algorithm's behavior is fixed (with high probability) before $e^{*}$ is planted, and a good event will say that $e^{*}$ avoids not just $H$ but everything $H$ could re-route through.

\begin{lemma}\label{lem:ck-generic}
    Let $G$ be a graph with girth $\geq 2k'+1$ and $m \coloneqq |E(G)| = \Omega(n^{1+1/k})$, let $k \geq 2$, $c > 1$, and let $\calA$ be a deterministic algorithm with budget $\Phi$. If $\Phi k = o(n^{1+1/k})$, then with probability $1 - o(1)$ over the uniform choice of $e^{*}$, $\calA$ incurs distortion $t \geq \Lambda = \Omega(ck)$ on the pair $(u^{*}, v^{*})$.
\end{lemma}

\begin{proof}
    Let a generic $c$-BANN oracle answers a query $(Q,X,c)$ as follows: For a query $q \in Q$ with $\min_{x \in X} \sfd_G(q,x) \geq 2$, it returns the point of $X$ closest to $q$ in $G$, under the same fixed tie-breaking rule as in Section~\ref{sec:instance-and-policy}, and reports distance $\min\{jc, \Lambda\}$; a query $q$ with $\min_{x \in X} \sfd_G(q,x) = 1$ returns the first neighbor in the fixed order and reports distance $c$; a distance query on a pair $j$ edges apart reports $\min\{jc, \Lambda\}$. This is exactly the policy of Section~\ref{sec:instance-and-policy} evaluated as though every edge weighed $c$. 

    Run the spanner algorithm $\calA$ with the generic $c$-BANN oracle and the distance oracle. The generic run is a fixed, special-edge-free computation and $\calA$ produces a spanner $H$ with at most $\Phi$ edges. The run also generates a transcript that contains all outputs by the oracles. Every item of the transcript therefore names a pair of points: the query point together with the answer point. Collect them and denote the set of pairs by
    \[
        T \;:=\; \big\{\, \text{the pair named by an item of the transcript} \,\big\} , \qquad |T| \;\leq\; \Phi.
    \]
    Since $E_H \subseteq T$ and $S$ is monotone, so 
    \[
        S(H) \;\subseteq\; S(T) ;
    \]
    and by Observation~\ref{obs:S-coverage}, so
    \[
        |S(T)| \;\leq\; k' \, |T| \;\leq\; \Phi k' .
    \]

    The set $S(T)$ is fixed by $G$ and algorithm $\calA$, before the special edge is drawn. Now place the special edge: Sample $e^{*}$ uniformly among the $m = \Omega(n^{1+1/k})$ edges (a larger $m$ only helps). Since $S(T)$ was fixed in advance and holds at most $\Phi k'$ of the $m$ edges,
    \[
        \Pr\big[\, e^{*} \notin S(T) \,\big] \;\geq\; 1 - \frac{\Phi k'}{m} \;=\; 1 - o(1),
    \]
    using $\Phi k' \leq \Phi k = o(m)$. Call $e^{*} \notin S(T)$ the good event, and assume it from now on. Since $S(H) \subseteq S(T)$, it also gives $e^{*} \notin S(H)$.

    We now argue that a true run of $\calA$ against the policy of Section~\ref{sec:instance-and-policy} produces exactly the same spanner as the generic run assuming the good event. Argue by induction: if everything of the generic run and the true run so far has agreed, the deterministic $\calA$ is in the same state and issues the same query. Then the oracles of both policies name the same answer point for the query point, selecting by graph distance in $G$ under the same tie-breaking order, since the only differences between the two policies are (1) the $c$-BANN oracle of the true run would avoid outputting pair $u^*,v^*$; (2) If the oracles output a pair whose shortest path $G$ contains $e^*$, then the output distances of the two policies are different. The good event guarantees that $e^{*}$, and every \emph{close} pair whose shortest path contains $e^{*}$, is never learned by the algorithm, which settles all queried pairs at hop distance at most $k'$. A queried pair at hop distance $j > k'$ whose shortest path uses $e^{*}$ is not covered by the good event, but it is covered by the truncation: the generic policy reports $\min\{jc, \Lambda\} = \Lambda$, while the true distance is at least $(j-1)c + 1 \geq k'c + 1 > \Lambda$ and so is also reported as $\Lambda$. The two cases are exhaustive, and the boundary is tight --- at $j = k'$ with $e^{*}$ on the path the true distance $(k'-1)c + 1$ is below $\Lambda$, so there it is the good event and not the truncation that saves the argument; this is exactly why $\Lambda$ is set to $k'c$. And $\calA$ adds the same edges to the output spanner $H$ according to the outputs of the oracles. Therefore, the transcripts $T$, $S(T)$, and output spanners of the generic and the true run are the same.

    We are ready to proceed to the conclusion. First, assuming the good event, the special edge is not in the output: were $e^{*}$ an edge of $H$, it would be a close pair, so its shortest path (the edge $e^{*}$ itself) would lie in $S(H) \subseteq S(T)$, which the good event forbids. Second, suppose for contradiction that $\sfd_H(u^{*}, v^{*}) < \Lambda$. The hypothesis of Lemma~\ref{lem:re-routing} is met, so it produces a pair of $H$ at most $k'$ edges apart whose shortest path runs through $e^{*}$; that is exactly the statement $e^{*} \in S(H)$, and $S(H) \subseteq S(T)$, again contradicting the good event. Therefore
    \[
        \sfd_H(u^{*}, v^{*}) \;\geq\; \Lambda \;=\; k'c \;\geq\; kc/2 .
    \]
\end{proof}

With Lemma~\ref{lem:re-routing} and Lemma~\ref{lem:ck-generic}, we are ready to show the lower bound for deterministic algorithms. To extend the lower bound to randomized algorithms, we 
will apply Yao's Principle.

\begin{lemma}[Yao's principle] \label{lem:yao}
    Let $\{I_\omega\}_{\omega \in \Omega}$ be a finite family of instances, each consisting of a metric on $[n]$ \emph{together with} a fixed rule by which the two oracles answer, and let $\calD$ be a distribution on $\Omega$. Write $t(\calA, I)$ for the distortion of the output of an algorithm $\calA$ run on instance $I$. Suppose that
    \[
        \E_{\omega \sim \calD}\big[\, t(\calA, I_\omega) \,\big] \;\geq\; \beta
        \qquad \text{for every \emph{deterministic} algorithm $\calA$ of budget $\Phi$.}
    \]
    Then for every randomized algorithm $\calR$ of budget $\Phi$, meaning one that respects the budget $\Phi$ for every random string, there is a single $\omega \in \Omega$ with $\E\big[t(\calR, I_\omega)\big] \geq \beta$, the expectation being over $\calR$'s coins alone.
\end{lemma}

\begin{proof}
    A randomized algorithm is a distribution over deterministic ones: fixing its random string $r$ leaves a deterministic algorithm $\calR_r$, and $\calR_r$ still has budget $\Phi$ because $\calR$ respects the budget for every random string. Moreover the deterministic algorithms of budget $\Phi$ form a \emph{finite} set here: each reads at most $\Phi$ entries, query and target sets are subsets of $[n]$, and in the family below the reported distances range over a finite set, so the whole transcript space is finite and such an algorithm is a finite decision tree. Every expectation below is therefore a finite sum. Now
    \[
        \max_{\omega \in \Omega} \E_r\big[t(\calR_r, I_\omega)\big]
        \;\geq\; \E_{\omega \sim \calD}\,\E_r\big[t(\calR_r, I_\omega)\big]
        \;=\; \E_r\,\E_{\omega \sim \calD}\big[t(\calR_r, I_\omega)\big]
        \;\geq\; \beta,
    \]
    where the first step holds because a maximum is at least an average, the second exchanges two finite sums, and the third applies the hypothesis to each fixed $\calR_r$. Any $\omega$ attaining the maximum proves the claim. 
    
\end{proof}

Now we are ready to prove Theorem~$\ref{thm: ck-lower-bound}$. 
\begin{proof}[Proof of Theorem~\ref{thm: ck-lower-bound}]
    Take the graph $G$ of Lemma~\ref{lem:girth-lower}, of girth $\geq k + 2 \geq 2k'+1$ and edges $m = \Omega(n^{1+1/k})$. The one condition of Lemma~\ref{lem:ck-generic} is exactly the hypothesis on the budget:
    \[
        \Phi k \;=\; o\!\left( \frac{n^{1+1/k}}{k} \cdot k \right) \;=\; o\big(n^{1+1/k}\big) \;=\; o(m).
    \]
    So every deterministic algorithm has $t \geq k'c \geq kc/2 = \Omega(ck)$ with probability $1 - o(1)$ over the randomness of $e^{*}$, and in particular on some instance, which proves the deterministic case. 

    For randomized algorithms, apply Lemma~\ref{lem:yao} to the uniform distribution over the family $\{I(e^{*})\}_{e^{*} \in E(G)}$, where each instance $I(e^*)$ consists of the metric $\sfd$ together with the answering policy of Section~\ref{sec:instance-and-policy} which is a fixed function of the query set $Q$, the target set $X$, and $e^{*}$. Since distortion is always at least $1$, Lemma~\ref{lem:ck-generic} gives, for every deterministic $\calA$ of budget $\Phi$,
    \[
        \E_{e^{*}}\big[t(\calA, I(e^{*}))\big] \;\geq\; k'c \cdot \big(1 - o(1)\big) \;=\; \Omega(ck),
    \]
    so Lemma~\ref{lem:yao} applies with $\beta = \Omega(ck)$. It follows that any randomized algorithm $\calR$ with budget $\Phi$ must incur expected distortion $\Omega(ck)$ on at least one instance in the family. Equivalently, for every such randomized algorithm $\calR$, there exists some choice of the special edge $e^{*} \in E(G)$ for which
    \[ \E_{\calR}\big[t(\calR,I(e^{*}))\big] \;=\; \Omega(ck). \]
    This proves the randomized case and completes the proof of Theorem~\ref{thm: ck-lower-bound}.

\end{proof}

\section{Application to $\sfW_q$ Distance Approximation} \label{sec: wasserstein-application}

In this section we prove Corollary \ref{cor: wasserstein-application} using a relatively standard reduction from bounded-hop spanner constructions to $\sfW_q$ approximation using minimum cost flow. We will be concerned with the algorithmic task of not only Wasserstein distance approximation, but rather the task of reporting a coupling $\gamma$ of distributions $\mu, \nu$ whose cost is a good multiplicative approximation of $\sfW_{q}(\mu, \nu)$.

We first briefly outline the connection between the problems of Minimum Cost Flow and Wasserstein distance computation. We then state and prove useful lemmas, and finally, we present the proof of Corollary \ref{cor: wasserstein-application}.

\begin{definition}[Minimum Cost Flow]
    Let $G = ([n], \vec{E})$ be a directed graph with $m := |\vec{E}|$, integral edge costs $c_e \in \{-C, \ldots, C\}$, integral capacity upper and lower bounds $u_e^{+}, u_e^{-} \in \{0, \ldots, U\}$ with $u_e^{-} \leq u_e^{+}$, and integral vertex demands $d \in \Z^n$. Then, the minimum cost flow problem asks to solve the following optimization problem
    \begin{align*}
        &\min_{f \in \R^{m}} \sum_{e \in \vec{E}} f_e \cdot c_e \\
        \text{Subject to } &\,\forall e \in \vec{E}, \, \, \, u_e^- \leq f_e \leq u_e^+ \\
        &\, \forall v \in [n], \, \, \, \sum_{u: (u, v) \in \vec{E}} f_{(u, v)} - \sum_{w: (v, w) \in \vec{E}} f_{(v, w)} = d_v \\
        &\sum_{v \in [n]} d_v = 0
    \end{align*}
\end{definition}

We now show how to reduce computation of $\sfW_q(\mu, \nu)$ over an arbitrary metric space $\calM$ to the problem of minimum cost flow. In what follows, we write $\mathrm{cost}(f) = \sum_{e} f_e c_e$. 



\begin{lemma} \label{lem: mcf-reduction}
    Let $\calM = ([n], \sfd)$ be a metric space with aspect ratio $\Delta$, $q \in [1, \infty)$, and $\mu, \nu \in \frac{1}{U} \Z^n$ denote discrete probability distributions over $[n]$ for some $U \in \N$. Let undirected weighted graph $H = ([n], E_H)$ be a $k$-hop $\sfD$-spanner of $\calM$, with the property that edge $(i,j) \in E_H$ has weight $\sfd(i,j)$. We let $\sfd_H$ denote the shortest path metric in $H$ and define $\sfd_{min} := \min \{ \sfd(i, j) : (i, j) \in E_H, i \neq j\}$. Then, consider the following minimum cost flow instance over directed graph $G = ([n] \times \{0, \ldots, k\}, \vec{E})$ with edge set $\vec{E}$ defined as follows:
    \begin{itemize}
        \item For each $(i, j) \in E_H$ and $\ell \in \{0, \ldots, k-1\}$, insert edges $(i, \ell) \to (j, \ell+1)$ and $(j, \ell) \to (i, \ell+1)$ into $\vec{E}$, each with cost $\lceil \sfd(i, j)^q / \sfd_{min}^q \rceil$ and capacity $U$. 
        \item For each $i \in [n]$ and $\ell \in \{0, \ldots, k-1\}$, insert edge $(i, \ell) \to (i, \ell+1)$ into $\vec{E}$, with cost $0$ and capacity $U$. 
        \item For each $i \in [n]$, set $d_{(i, 0)} := -U \cdot \mu(i)$ and $d_{(i, k)} := U \cdot \nu(i)$. For all $\ell \in \{1, \ldots, k-1\}$, set $d_{(i, \ell)} = 0$. 
    \end{itemize}
    Then the instance is feasible, and the optimal flow satisfies 
    \begin{align*}
        k^{1-q} \cdot \sfW_q(\mu, \nu)^q \, \leq \, \frac{\sfd_{min}^q}{U} \cdot \mathrm{cost}(f^\star) \, \leq \, 2 \sfD^q \cdot \sfW_q(\mu, \nu)^q.
    \end{align*}
    Moreover, from any feasible flow $f$ one can extract in $O(|\vec{E}| \cdot k + n)$ additional time a coupling $\gamma_f \in \Gamma(\mu, \nu)$, supported on at most $|\vec{E}|$ pairs, with $$\sfW_q(\mu, \nu)^q \leq \sum_{i,j} \gamma_f(i,j) \sfd(i,j)^q \leq 2k^{q-1}\sfD^q \sfW_q(\mu, \nu)^q$$
    
\end{lemma}

We defer the proof of Lemma \ref{lem: mcf-reduction} to section \ref{sec: mcf-reduction}. Given Lemma \ref{lem: mcf-reduction}, we can conclude the statement of Corollary \ref{cor: wasserstein-application} by utilizing fast algorithms for min-cost flow. 

\begin{theorem}[Theorem 1.1 in \cite{CKLPPS22}] \label{thm: mcf}
    There is an algorithm that, on a graph $G = (V, E)$ with $m$ edges, vertex demands, upper/lower edge capacities, and edge costs, all integral with capacities bounded by $U$ and costs bounded by $C$, computes an exact min-cost flow in $m^{1+o(1)} \log U \log C$ time with high probability.
\end{theorem}

\begin{proof}[Proof of Corollary \ref{cor: wasserstein-application}]
    We let spanner $H$ be constructed via Algorithm \ref{alg: fast-spanners} and condition on its correctness. By Theorem \ref{thm: upper-bound}, $H$ is a $(4k-1)$-hop $\sfD$-spanner, where $\sfD = O(c k)$. Observe that the graph $G$ constructed in the statement of Lemma \ref{lem: mcf-reduction} has $n(k+1)$ vertices and $|\vec{E}| = k(2|E_H| + n)$ directed edges, all capacities are bounded by $U$, and all costs are non-negative and bounded by
    \begin{align*}
        C = \max_{(i,j) \in E_H} \left\lceil \frac{\sfd(i,j)^q}{\sfd_{min}^q} \right\rceil \leq \lceil \Delta^q \rceil,
    \end{align*}
    Combining Theorem \ref{thm: mcf} with Lemma \ref{lem: mcf-reduction} and the $O(|\vec{E}| \cdot k + n)$-time coupling extraction, there exists an algorithm that reports a coupling $\tilde{\gamma}$ of $\mu, \nu$ such that 
    \begin{align*}
        \sfW_q(\mu, \nu)^q \leq \sum_{i,j} \tilde{\gamma}(i,j) \sfd(i,j)^q \leq 2k^{q-1} \sfD^q \cdot \sfW_q(\mu, \nu)^q
    \end{align*}

    Therefore, the algorithm returns $\tilde{\gamma}$ whose cost is an $O(ck^{2 - 1/q})$ approximation of $\sfW_q(\mu, \nu)$. Union bounding over the failure probabilities of Algorithm \ref{alg: fast-spanners} and the algorithm of Theorem \ref{thm: mcf} \cite{CKLPPS22}, we can conclude that the $\sfW_q$ approximation algorithm succeeds with high probability. 

    \paragraph{Runtime.} By Theorem \ref{thm: upper-bound}, the runtime of Algorithm \ref{alg: fast-spanners} is $O(\bnntime(n) \cdot kn^{1/k} \log \Delta)$ and $|E_H| = O(kn^{1+1/2k} \log \Delta)$. Constructing the minimum-cost flow instance takes time $O(|E_H| k + kn) = O(k^2 n^{1+1/2k} \log \Delta)$. By Theorem \ref{thm: mcf}, computing the optimal flow in graph $G$ takes time $|\vec{E}|^{1 +o(1)} \log U \log C = q \cdot (k^2 n^{1 + 1/2k})^{1 + o(1)} \log^{2 + o(1)} \Delta \log U $.

    When $q \leq \log U$, computing the optimal flow takes time $(k^2 n^{1 + 1/2k})^{1 + o(1)} \polylog \Delta \cdot \polylog U $. When $q > \log U$, Lemma \ref{lem: wq-sandwich} allows us to conclude that $\sfW_{\log U}$ is a constant factor approximation of $\sfW_q$. Therefore, our algorithm may safely estimate $\sfW_{\log U}$ up to constant factors in the distortion. The runtime of computing the optimal flow is then also $(k^2 n^{1 + 1/2k})^{1 + o(1)} \polylog \Delta \cdot \polylog U $.

    Thus, the final runtime of the $\sfW_q$ approximation algorithm is $O(\bnntime(n) \cdot kn^{1/k} \log \Delta) + (k^2 n^{1 + 1/2k})^{1 + o(1)} \polylog \Delta \cdot \polylog U$.
\end{proof}

\subsection{Proof of Lemma \ref{lem: mcf-reduction}} \label{sec: mcf-reduction}





        


                

To prove Lemma \ref{lem: mcf-reduction}, we prove a few helpful statements first. We begin by showing that every feasible flow decomposes into paths from layer $0$ to layer $k$, which is the mechanism by which we extract a coupling of $\mu$ and $\nu$ from a flow.
 
\begin{lemma} \label{lem: flow-decomposition}
    Let $f$ be any feasible flow in the Minimum Cost Flow instance of Lemma \ref{lem: mcf-reduction}. There exist $J \leq |\vec{E}|$ many scalars $\lambda_1, \ldots, \lambda_J > 0$ and corresponding directed paths $P_1, \ldots, P_J$ from $(x_t, 0)$ to $(y_t, k)$ such that $f = \sum_{t=1}^{J} \lambda_t \mathbf{1}_{P_t}$. Moreover, this decomposition can be computed in $O(k |\vec{E}| + n)$ time.
\end{lemma}
 
\begin{proof}
    Every directed edge of $\vec{E}$ leads from layer $\ell$ to layer $\ell + 1$ for some $\ell \in \{0, \ldots, k-1\}$, so $G$ is acyclic and every directed path in $G$ has at most $k$ directed edges. Moreover, every vertex $(v, \ell)$ with $\ell \in \{1, \ldots, k-1\}$ has demand $0$, so its in-flow equals its out-flow, and hence it has an incoming directed edge carrying positive flow if and only if it has an outgoing directed edge carrying positive flow.
 
    Consider any edge $e = (u, 0) \to (v, 1) $ for which $f_e > 0$. Walking forwards from $(v, 1)$ along directed edges carrying positive flow, the above observation guarantees that we may continue until layer $k$ is reached. This produces a directed path $P$ from some $(x, 0)$ to some $(y, k)$ with $e \in P$; we set $\lambda_1 := \min_{e' \in P} f_{e'} > 0$ and replace $f$ by $f - \lambda_1 \mathbf{1}_{P}$. 
    
    The result is again a non-negative flow which conserves flow at every vertex of layers $1, \ldots, k-1$, so the argument may be recursively applied. Each repetition sets $f_{e'} = 0$ for at least one directed edge $e'$. Moreover, the flow along every edge is non-increasing over the course of this procedure. Thus, after $J \leq |\vec{E}|$ repetitions we arrive at the zero flow, which yields the claimed decomposition. 
    
    \paragraph{Runtime.} Maintaining a pointer that scans $\vec{E}$ once for the directed edges carrying positive flow takes $O(|\vec{E}| + n)$ time in total, and each of the $J$ repetitions takes $O(k)$ time, so the decomposition is computed in $O(k|\vec{E}| + n)$ time.
\end{proof}

Such a decomposition of a feasible flow can be used to construct a corresponding coupling between the distributions $\mu, \nu$. 

\begin{lemma} \label{lem: flow-coupling}
    Let $f$ be any feasible flow in the Minimum Cost Flow instance of Lemma \ref{lem: mcf-reduction}, let $\lambda_1, \ldots, \lambda_J$ and $P_1, \ldots, P_J$ be a decomposition of $f$ as in Lemma \ref{lem: flow-decomposition}, and define $\gamma_f \in \R_{\geq 0}^{n \times n}$ by
    \begin{align*}
        \gamma_f(x, y) \; := \; \frac{1}{U} \sum_{t \in [J] \, : \, (x_t, y_t) = (x, y)} \lambda_t.
    \end{align*}
    Then, $\gamma_f \in \Gamma(\mu, \nu)$. Moreover, $\gamma_f$ is supported on at most $J \leq |\vec{E}|$ pairs, and is specified by the list $(x_t, y_t, \lambda_t / U)_{t \in [J]}$ returned by Lemma \ref{lem: flow-decomposition}.
\end{lemma}
\begin{proof}
    As no directed edge of $\vec{E}$ enters layer $0$ or leaves layer $k$, the demand constraints at $(x, 0)$ and at $(y, k)$ state that the out-flow at $(x, 0)$ is $U \mu(x)$ and that the in-flow at $(y, k)$ is $U \nu(y)$. Every path $P_t$ leaves layer $0$ only at $(x_t, 0)$ and enters layer $k$ only at $(y_t, k)$, so it contributes exactly $\lambda_t$ to the former and exactly $\lambda_t$ to the latter. As $f = \sum_{t \in [J]} \lambda_t \mathbf{1}_{P_t}$, we conclude that for every $x, y \in [n]$,
    \begin{align*}
        \sum_{t \in [J] \, : \, x_t = x} \lambda_t \; = \; U \mu(x) \qquad \text{and} \qquad \sum_{t \in [J] \, : \, y_t = y} \lambda_t \; = \; U \nu(y).
    \end{align*}
    By the definition of $\gamma_f$, the two left hand sides are exactly $U \sum_{y \in [n]} \gamma_f(x, y)$ and $U \sum_{x \in [n]} \gamma_f(x, y)$, so $\gamma_f$ is a non-negative matrix with marginals $\mu$ and $\nu$, i.e. $\gamma_f \in \Gamma(\mu, \nu)$. Finally, $\gamma_f(x, y)$ is non-zero only if $(x, y) = (x_t, y_t)$ for some $t \in [J]$, and $J \leq |\vec{E}|$ by Lemma \ref{lem: flow-decomposition}. 
\end{proof}

We now bound the cost of the optimal flow from above, which we do by exhibiting a feasible flow obtained from routing an optimal coupling along paths in the spanner $H$.
 
\begin{lemma} \label{lem: mcf-upper}
    Let $f^{\star}$ be the optimal flow in the Minimum Cost Flow instance of Lemma \ref{lem: mcf-reduction}. Then,
    \begin{align*}
        \mathrm{cost}(f^{\star}) \leq \; \frac{2\sfD^q U}{\sfd_{min}^q} \cdot \sfW_q(\mu,\nu)^q
    \end{align*}
\end{lemma}
\begin{proof}
    Let $\gamma^\star \in \Gamma(\mu, \nu)$ attain $\sfW_q(\mu, \nu)^q = \sum_{i, j \in [n]} \gamma^\star(i, j) \sfd(i,j)^q$. As $H$ is a $k$-hop spanner, for each $i, j \in [n]$ there is a path $P_{ij}$ from $i$ to $j$ in $H$ using at most $k$ edges and of weight $w(P_{ij}) := \sum_{(a, b) \in P_{ij}} \sfd(a,b) \leq \sfD \cdot \sfd(i,j)$, where $P_{ii}$ is the empty path. Let $\widetilde{P}_{ij}$ denote the directed path in $G$ from $(i, 0)$ to $(j, k)$ which traverses the edges of $P_{ij}$ in its first $|P_{ij}|$ directed edges and then remains at $j$ using the zero-cost directed edges, and define
    \begin{align*}
        f \; := \sum_{i, j \in [n]} U \gamma^\star(i,j) \cdot \mathbf{1}_{\widetilde{P}_{ij}}.
    \end{align*}
    Being a non-negative combination of paths from layer $0$ to layer $k$, $f$ conserves flow at every vertex of layers $1, \ldots, k-1$. The out-flow at $(i, 0)$ is $\sum_{j \in [n]} U \gamma^\star(i,j) = U \mu(i)$ and the in-flow at $(j, k)$ is $\sum_{i \in [n]} U \gamma^\star(i,j) = U \nu(j)$, as the demands require. Every directed edge carries at most $\sum_{i, j \in [n]} U \gamma^\star(i,j) = U$ units of flow, so $f$ also respects the capacities.
 
    It remains to bound $\mathrm{cost}(f)$. As $\sfd(a,b) \geq \sfd_{min}$ for every $(a,b) \in E_H$ with $a \neq b$, and as $\lceil x \rceil \leq x + 1 \leq 2x$ for every $x \geq 1$, each directed edge of $\vec{E}$ arising from an edge $(a,b) \in E_H$ has cost at most $2 \sfd(a,b)^q / \sfd_{min}^q$. Thus, 
    \begin{align*}
        \mathrm{cost}(f) \; &= \; U \sum_{i, j \in [n]} \gamma^\star(i,j) \sum_{(a,b) \in P_{ij}} \left\lceil \frac{\sfd(a,b)^q}{\sfd_{min}^q} \right\rceil \\
        &\leq \; \frac{2U}{\sfd_{min}^q} \sum_{i, j \in [n]} \gamma^\star(i,j) \sum_{(a,b) \in P_{ij}} \sfd(a,b)^q \\
        &\leq \; \frac{2U}{\sfd_{min}^q} \sum_{i, j \in [n]} \gamma^\star(i,j) \cdot \left(\sum_{(a, b) \in P_{ij}} \sfd(a, b) \right)^q  \\
        &\leq \; \frac{2 \sfD^q U}{\sfd_{min}^q} \sum_{i, j \in [n]} \gamma^\star(i,j) \sfd(i,j)^q \tag{$H$ is a $\sfD$-spanner}
    \end{align*}
    where the first equality uses that the directed edges of $\widetilde{P}_{ij}$ which do not arise from an edge of $P_{ij}$ have cost $0$. Since $\mathrm{cost}(f^{\star}) \leq \mathrm{cost}(f)$, the claim follows. 
\end{proof}
 
In the other direction, we lower bound the cost of an arbitrary feasible flow by the cost of the coupling it induces. 
 
\begin{lemma} \label{lem: mcf-lower}
    Let $f$ be any feasible flow in the Minimum Cost Flow instance of Lemma \ref{lem: mcf-reduction}. Then,
    \begin{align*}
        \mathrm{cost}(f) \geq \;  \frac{U k^{1-q}}{\sfd_{min}^q} \cdot \sum_{i, j \in [n]} \gamma_f(i, j) \sfd(i, j)^q
    \end{align*}
\end{lemma}
\begin{proof}
    Let $f = \sum_{t \in [J]} \lambda_t \mathbf{1}_{P_t}$ be the decomposition of Lemma \ref{lem: flow-decomposition} and write $c(P) := \sum_{e \in P} c_e$, so that $\mathrm{cost}(f) = \sum_{t \in [J]} \lambda_t \cdot c(P_t)$. Fix $t \in [J]$ and let $x_t = v_0, v_1, \ldots, v_k = y_t$ be the vertices of $[n]$ along the path $P_t$. Then, 
    \begin{align*}
        \sfd(x_t, y_t)^q \; \leq \; \Big( \sum_{h = 0}^{k-1} \sfd(v_{h}, v_{h+1}) \Big)^q \; &\leq \; k^{q-1} \sum_{h = 0}^{k-1} \sfd(v_{h}, v_{h+1})^q \tag{Jensen's Inequality} \\
        &\leq \; k^{q-1} \sfd_{min}^q \sum_{h = 0}^{k-1}  \left\lceil \frac{\sfd(v_{h}, v_{h+1})^q}{\sfd_{min}^q} \right\rceil \; \\
        &= \; k^{q-1} \sfd_{min}^q \cdot c(P_t)
    \end{align*}
    Summing over $t \in [J]$ and recalling the definition of $\gamma_f$,
    \begin{align*}
        \sum_{i, j \in [n]} \gamma_f(i,j) \sfd(i,j)^q \; = \; \frac{1}{U} \sum_{t \in [J]} \lambda_t \cdot \sfd(x_t, y_t)^q \; &\leq \; \frac{k^{q-1} \sfd_{min}^q}{U} \sum_{t \in [J]} \lambda_t \cdot c(P_t) \\
        &= \; \frac{k^{q-1}\sfd_{min}^q}{U} \cdot \mathrm{cost}(f),
    \end{align*}
    which rearranges to the claim.
\end{proof}
 
We are now ready to prove the statement of Lemma \ref{lem: mcf-reduction}.
 
\begin{proof}[Proof of Lemma \ref{lem: mcf-reduction}]
    From Lemma \ref{lem: mcf-upper}, we can conclude that
    \begin{align*}
        \frac{\sfd_{min}^q}{U} \cdot \mathrm{cost}(f^{\star}) \leq 2\sfD^q \sfW_q(\mu, \nu)^q
    \end{align*}
    In the other direction, we apply Lemma \ref{lem: mcf-lower} to $f^\star$ and use that $\gamma_{f^\star} \in \Gamma(\mu, \nu)$, which gives
    \begin{align*}
        \frac{\sfd_{min}^q}{U} \cdot \mathrm{cost}(f^\star) \; \geq \; k^{1-q} \sum_{i, j \in [n]} \gamma_{f^\star}(i,j) \sfd(i,j)^q \; \geq \; k^{1-q} \cdot \sfW_q(\mu, \nu)^q.
    \end{align*}
    Finally, Lemmas \ref{lem: flow-decomposition} and \ref{lem: flow-coupling} produce, in $O(k|\vec{E}| + n)$ time, a coupling $\gamma_{f^{\star}} \in \Gamma(\mu, \nu)$ supported on at most $|\vec{E}|$ pairs such that 
    \begin{align} \label{eq: coupling-guarantee}
        \sfW_q(\mu, \nu)^q \; \leq \; \sum_{i, j \in [n]} \gamma_{f^\star}(i,j) \sfd(i,j)^q \; \leq \; 2 k^{q-1} \sfD^q \cdot \sfW_q(\mu, \nu)^q
    \end{align}
    so that $\gamma_{f^{\star}}$ is an $O(k^{1- 1/q} \sfD)$-approximate coupling of $\mu, \nu$ for $\sfW_q$.
\end{proof}

\paragraph{AI Disclosure.} ChatGPT-5.6 Pro and GPT6-Astra were used to assist with the discovery and proofs of Theorems \ref{thm: upper-bound} and \ref{thm: lower-bound}. The results emerged through a long back and forth conversation between the authors and the model. The authors validated the proofs, extended the results, and have rewritten/substantially edited proofs to improve the presentation.

\bibliographystyle{alpha} 
\bibliography{krish} 


\end{document}